\documentclass[twoside]{article} % ,twocolumn
\usepackage[a4paper]{geometry}
\usepackage[utf8]{inputenc} % ou \usepackage[latin1]{inputenc}
\usepackage[T1]{fontenc} % ou \usepackage[OT1]{fontenc}
\usepackage{RR}

\usepackage{todonotes,stmaryrd}
\usepackage{cite}
\usepackage{xspace}
\usepackage{amssymb,amsmath,mathtools,amsthm}
\usepackage{url}
\usepackage{algorithm}
\usepackage[noend]{algorithmic}
\usepackage{float}
\usepackage{subfigure}

\usepackage{enumitem}
\usepackage{multirow}

\usepackage{hyperref}

\graphicspath{{figs/}}

\newtheorem{theorem}{Theorem}
\newif\ifdraft\drafttrue

\ifdraft
\newcommand\todos[1]{\todo[inline]{TODO (all): #1}}
\newcommand\ma[1]{\todo[color=green!40,inline]{TODO (Mathieu): #1}}

\newcommand\gb[1]{\todo[color=yellow!40,inline]{TODO (Guillaume): #1}}
\newcommand\rb[1]{\todo[color=blue!40,inline]{TODO (Razieh): #1}}

\else
\newcommand\todos[1]{}
\newcommand\ma[1]{}

\newcommand\gb[1]{}
\newcommand\rb[1]{}

\fi
\newcommand{\ie}{\emph{i.e.}\xspace}
\newcommand{\eg}{\emph{e.g.}\xspace}

\newcommand{\etal}{\emph{et al.}\xspace}

\newcommand{\myexample}[1]{	\emph{\small Example: #1}}
\renewcommand{\myexample}[1]{\emph{Example.} #1}
\newtheorem{mydef}{Definition}
\newcommand{\feature}[1]{\emph{#1}}

\RRdate{February 2015}
\RRauthor{% les auteurs
Guillaume B\'{e}can\thanks[i1]{Inria/IRISA, University of Rennes 1, France}
\and Razieh Behjati\thanks[i2]{SIMULA, Norway}
\and Arnaud Gotlieb\thanksref{i2}
\and Mathieu Acher\thanksref{i1}
}
\authorhead{B\'{e}can, Behjati, Gotlieb and Acher}
\RRtitle{Synthèse de Modèles de Caractéristiques Attribués à Partir de Descriptions de Produits : Fondations} % French title
\RRetitle{Synthesis of Attributed Feature Models From Product Descriptions: Foundations}
\titlehead{Synthesis of Attributed Feature Models (Foundations)}
\RRresume{%
La modélisation de caractéristiques (features) est un formalisme largement utilisé pour décrire un ensemble de produits (ou configurations).
Comme une élaboration manuelle est une tâche longue et ardue, de nombreuses techniques ont été proposées pour extraire des modèles de caractéristiques à partir de types variés d'artefacts. Aucune d'elles cependant ne considère la synthèse d'attributs de caractéristiques (ou de contraintes portant sur les attributs) malgré l'utilité pratique des attributs pour documenter différentes valeurs d'un ensemble de produits. 
Dans ce rapport, nous développons un algorithme pour synthétiser des modèles de caractéristiques attribués, étant donné un ensemble de descriptions de produits.
Nous présentons des techniques correctes, complètes et paramétrables pour calculer toutes les hiérarchies possibles, groupes de caractéristiques, placements d'attributs, valeurs de domaines et contraintes.
Nous effectuons une analyse de complexité par rapport au nombre de caractéristiques, d'attributs, de configurations et de la taille des domaines. Nous évaluons aussi le passage à l'échelle de notre procédure de synthèse en utilisant des matrices de configurations générées aléatoirement.
Ce rapport est une première étape qui vise à décrire les fondations pour synthétiser des modèles de caractéristiques attribués.
} % French abstract

\RRabstract{%
Feature modeling is a widely used formalism to characterize a set of products (also called configurations). 
 As a manual elaboration is a long and arduous task, numerous techniques have been proposed to reverse engineer feature models from various kinds of artefacts. But none of them synthesize feature attributes (or constraints over attributes) despite the practical relevance of attributes for documenting the different values across a range of products.
 In this report, we develop an algorithm for synthesizing attributed feature models given a set of product descriptions. 
We present sound, complete, and parametrizable techniques for computing all possible hierarchies, feature groups, placements of feature attributes, domain values, and constraints. 
 We perform a complexity analysis w.r.t. number of features, attributes, configurations, and domain size. We also evaluate the scalability of our synthesis procedure using randomized configuration matrices. 
 This report is a first step that aims to describe the foundations for synthesizing attributed feature models. 
} % English abstract

\RRmotcle{Modèle de caractéristiques, Description de produits, Rétro-ingénierie} % French keywords
\RRkeyword{Attributed feature models, Product descriptions, Reverse engineering} % English keywords
\RRprojets{DiverSE and SIMULA} % Projects
 \RCRennes % Rennes - Bretagne Atlantique
\begin{document}
\RRNo{8680}
\makeRR   % cas d'un rapport de recherche
%\makeRT % cas d'un rapport technique.
%% a partir d'ici, chacun fait comme il le souhaite

%%%%%%%%%%%%%%%%%%%%%%%%%%%%%%%%%%%%%%%%%%%

\section{Introduction}
Configuration options are ubiquitous. They allow users to customize or choose their product.
Options (also referred as \emph{features} or \emph{attributes}) refer to functional and non-functional aspects of a system, at different level of granularity -- from parameters in a function to a whole service. 
For example, software practitioners can activate or deactivate some functionalities and tune the energy consumption and memory footprint when building a product for embedded systems. Customers themselves can select a set of desired options to match their requirements -- for satisfying their functional needs without e.g. reaching a maximum budget.  
% 

% Parameter, characteristic, functionality, feature, or attribute are other alternate terms commonly employed to refer to options.  
Modeling features or attributes of a given set of products is a crucial activity in \emph{Software Product Line (SPL)} engineering. 
The formalism of \emph{feature models} (FMs) is widely employed for this purpose~\cite{czarnecki2000,czarnecki2006,czarnecki2007,ACMSurveyHubaux}. FMs delimit the scope of a family of related products (i.e., an SPL) and formally document what combinations of features are supported. Once specified, FMs can be used for model checking an SPL~\cite{classen2010,cordy2013,beek2013,thuem2014}, automating product configuration~\cite{ebrahim2012,ebrahim2013,hubaux2013Book,DBLP:conf/kbse/GuoZORCAA14}, computing relevant information~\cite{benavides2010} or communicating with stakeholders~\cite{berger2014}. In many generative or feature-oriented approaches, FMs are also central for deriving software-intensive products~\cite{apel2013book}.  

% Since their introduction in 1990 by Kang \etal, some extensions of the formalism of FMs have been proposed. 
\emph{Feature attributes} are a useful extension, intensively employed in practice, for documenting the different values across a range of products~\cite{hubaux2010}. 
%  intensively employed in practice
 With the addition of attributes, optional behaviour can be made dependent not
only on the presence or absence of features, but also on the satisfaction
of constraints over domain values of attributes~\cite{cordy2013}. Recently, languages and tools have emerged to fully support attributes in feature modeling and SPL engineering (e.g., see~\cite{benavides2010,DBLP:conf/splc/EichelbergerS13,TVL,bak2010,seidl2014capturing,SPLConqueror,siegmundKKABRS12,siegmund2013,guo2013, kolesnikovASSKS13,guo2014,DBLP:conf/kbse/GuoZORCAA14,alferez:hal-01023159}). 

% to model and reason about a set of products (also called configurations)
% ....

% Some extensions of feature models have been proposed.
% Very early, Kang \etal use an example of attribute in~\cite{kang1990} while
% Czarnecki \etal coin the term "feature attribute" in~\cite{czarnecki2002}.
% However, as
% reported in~\cite{benavides2010}, the vast majority of research in feature modeling has
% focused on basic, propositional feature models~\cite{czarnecki2006,czarnecki2007}.

% In the area of SPL engineering, the idea of extending (or augmenting) feature models with
% quality attributes was proposed by many~\cite{zhang2011,white2014,guo2011}.
% \emph{Attributed Feature Models} (AFMs) have emerged... ~\cite{bak2010,cordy2013,joseISSTA14}
% A recent survey showed that engineers commonly need these constructs~\cite{hubaux2010b}
% CVL provides such constructs ~\cite{berger2013}
% ~\cite{czarnecki2006}
% Clafer, TVL, FAMA, etc.

The manual elaboration of a feature model -- being with attributes or not -- is known to be a daunting and error-prone task~\cite{FSE2013,AcherVAMOS12,acher-cleve-etal:2013,she2011,ryssel2011,bagheri2012Wikipedia,FerrariSd13,alves2008,niu2009,rashid2009,chen2005,nadi2014,czarnecki2007,she2011,janota2008,AcherVaMoS13,lopez2011,Lopez-HerrejonGBSE12,herejonFASE2013,andersen2012,she2008,she2014efficient,AcherVaMoS13,becanESE2015,LopezHerrejon2014}. 
The number of features, attributes, and dependencies among them can be very important so that practitioners can face severe difficulties for accurately modeling a set of products. % It is easy to forget a constraint
In response, numerous \emph{synthesis} techniques have been developed for synthesising feature models~\cite{czarnecki2007,she2011,janota2008,AcherVaMoS13,lopez2011,Lopez-HerrejonGBSE12,herejonFASE2013,andersen2012,she2008,she2014efficient,AcherVaMoS13,becanESE2015,LopezHerrejon2014}.
 Until now, the impressive research effort has focused on synthesizing basic, Boolean feature models -- without feature attributes. Despite the evident opportunity of encoding quantitative information as attributes, the synthesis of attributed feature models has not yet caught attention.  
 % TODO: state why the problem of synthesising AFM is not amenable to synthesising (B)FM
 None of the existing techniques synthesize feature attributes, domain values or constraints over attributes.

 In this report, we develop the theoretical foundations and techniques for synthesizing attributed feature models given a set of product descriptions.
 We present sound, complete, and parametrizable techniques for computing hierarchies, feature groups, placements of feature attributes, domain values, and constraints. 
We describe algorithms for computing logical relations between features and attributes. % \gb{no CP techniques in this TR}
The synthesis is capable of taking knowledge (e.g., about the hierarchy and placement of attributes) into account so that users can specify, if needs be, a hierarchy or some placements of attributes.

 We perform a complexity analysis of our synthesis procedure with regards to the number of configurations, features, attributes, and domain values. 
 We also evaluate the scalability of the synthesis using randomized configuration matrices. 
%Our results show that the synthesis can scale for a large number of configurations, features, and domain % values. We characterize the theoretical limits of the formalism of attributed FMs and discuss the practical % impacts. 
 Our work both strengthens the understanding of the formalism and provides the basis of a tool-supported solution for synthesizing attributed feature models.

The foundations presented in this report open avenues for investigating novel reverse engineering scenarios involving attributes. 
Numerous works have developed techniques for mining and extracting features or constraints from various kinds of artefacts~\cite{FSE2013,AcherVAMOS12,acher-cleve-etal:2013,she2011,ryssel2011,bagheri2012Wikipedia,FerrariSd13,alves2008,niu2009,rashid2009,chen2005,nadi2014} (textual requirements, design models, source code, semi-structured or informal product descriptions, configurators, etc.). However, they do not support attributes despite the presence of non Boolean data in some of these artefacts. Such automated extraction techniques can be used to process different types of artefacts and eventually fed our synthesis algorithm. 

%In this paper, we do not consider such a broad view; we focus solely on the synthesis of AFMs. 

% \ma{opens avenues for reverse engineering new kinds of artefacts with variability}

%\ma{the following in the conclusion maybe}
%The contributions of the paper can be summarized as follows:
%\begin{itemize}
%\item We develop a synthesis algorithm... We show that the set of synthesized constraints is sound and complete by construction and polynomial in time. The algorithm can be parameterized... \ma{approximation and equivalence properties: also contribution};
%\item We empirically evaluate the scalability of the synthesis algorithm;
%\item We empirically investigate the overapproximation effect induced by the formalism of attributed feature model
%\end{itemize}

% \ma{Maybe we have to split the evaluation in two or three sections}
% Since~\cite{czarnecki2007},

The remainder of the report is organized as follows. 
Section~\ref{sec:related} discusses related work. % only consider boolean FMs. 
Section~\ref{sec:motivation} further motivates the need of synthesising attributed FMs.
Section~\ref{sec:problem} exposes the problem of synthesizing attributed FMs.
Section~\ref{sec:approach} presents our algorithm targeting this problem.
Section~\ref{sec:theoretical} and ~\ref{sec:practical} respectively evaluate the synthesis techniques from a theoretical and practical aspect.
In Section~\ref{sec:threats} we discuss threats to validity. Section~\ref{sec:conclusion} summarizes the contributions and describes future work.

\section{Related Work}
\label{sec:related}
Numerous works address the synthesis or extraction of FMs. 
Despite the availability of some tools and languages supporting attributes, no prior work consider the synthesis of attributed FMs. They solely focus on Boolean FMs.

% \ma{I propose we discuss related work upfront, very early in the paper, with a simple message: (1) no synthesis algorithms for attributed feature models (despite the availability of some tools/languages like Clafer, TVL, CVL, VM, FAMA for specifying/reasoning about attributed feature models); (2) most of the effort focused on Boolean feature models or address reverse engineering/mining scenarios amenable to the synthesis of Boolean feature models.}

% input
% output
% what techniques are used?

\subsection{Synthesis of Feature Models} 
Techniques for synthesising an FM from a set of dependencies (e.g., encoded as a propositional formula) or from a set of configurations (e.g., encoded in a product comparison matrix) have been proposed~\cite{czarnecki2007,she2011,janota2008,AcherVaMoS13,lopez2011,Lopez-HerrejonGBSE12,herejonFASE2013,andersen2012,she2008}.

In~\cite{czarnecki2007,andersen2012,she2014efficient}, the authors calculate a diagrammatic representation of all possible FMs from a propositional formula (CNF or DNF). 
In~\cite{AcherVaMoS13}, we propose a synthesis procedure that processes user-specified knowledge for organizing the hierarchy of features. 
In~\cite{becanESE2015}, we also propose a set of techniques for synthesizing FMs that are both correct w.r.t input propositional formula and present an appropriate hierarchy.
% These works are the basis of the algorithm presented in this paper.

The algorithms proposed in~\cite{lopez2011,herejonFASE2013,Lopez-HerrejonGBSE12,LopezHerrejon2014} take as input a set of configurations. The generated FM may not conform to the input configurations, that is, the FM may be an over-approximation of the configuration set. Our work aims to study whether similar properties arise in the context of attributed feature models. The approaches presented in~\cite{lopez2011,herejonFASE2013,Lopez-HerrejonGBSE12,LopezHerrejon2014} do not control the way the FM hierarchy is synthesized. 
The major drawback is that the resulting hierarchy is likely to be difficult to read, understand, and exploit~\cite{becanESE2015,WebFML2014}. In contrast, She \etal~\cite{she2011} proposed a heuristic to synthesize an FM presenting an appropriate hierarchy. Applied on the software projects Linux, eCos, and FreeBSD, the technique assumes the existence of feature descriptions. % \gb{no attributes for eCos?}
Janota \etal~\cite{janota2008} developed an interactive editor, based on logical techniques, to guide users in synthesizing an FM from a propositional formula.
In prior works~\cite{becanESE2015,WebFML2014,AcherVaMoS13}, we develop techniques for taking the so-called ontological semantics into account when synthesizing feature models. 
 Our work shares the goal of interactively supporting users -- this time in the context of synthesizing attributed feature models.

\textit{Overall, numerous works exist for the synthesis of FMs but none support attributes.}

\subsection{Extraction of Feature Models} 
\label{subsec:fmextraction}
Considering a broader view, reverse engineering techniques have been proposed to extract FMs from various artefacts. 

Davril \etal~\cite{FSE2013} presented a fully automated approach, based on prior work~\cite{TSEJane}, for constructing FMs from publicly available product descriptions found in online product repositories and marketing websites such as SoftPedia and CNET. %The proposal is evaluated in the anti-virus domain.
In~\cite{AcherVAMOS12}, a semi-automated procedure to support the transition from product descriptions (expressed in a tabular format) to FMs is proposed. 
In~\cite{acher-cleve-etal:2013}, architectural knowledge, plugins dependencies and the slicing operator are combined to obtain an exploitable and maintainable FM. 
Ryssel et al. developed methods based on Formal Concept Analysis and analyzed incidence matrices containing matching relations~\cite{ryssel2011}.

 % These works assume a certain structure in the products descriptions or a knowledge that is exploited to hierarchically organize the features. 
% In the operating system domain, She et al.~\cite{she2011} proposed a procedure to rank the correct parent features in order to reduce the task of a user.

% User support is not provided for refactoring the resulting FM

Yi \etal~\cite{yiRE2012} proposed to apply support vector machine and genetic techniques to mine binary constraints (requires and excludes) from Wikipedia. This scenario is particularly relevant when dealing with \emph{incomplete} dependencies. They evaluated their approach on two FMs of SPLOT. 
%Our techniques might be used to \emph{suggest} logical relationships in case strong ontological relations between features are inferred. 
 Bagheri \etal~\cite{bagheri2012Wikipedia} proposed a collaborative process to mine and organize features using a combination of natural language processing techniques and Wordnet. Ferrari \etal~\cite{FerrariSd13} applied natural language processing techniques to mine commonalities and variabilities from brochures. 
Alves \etal~\cite{alves2008}, Niu \etal~\cite{niu2009}, Weston \etal~\cite{rashid2009} and Chen \etal~\cite{chen2005} applied information retrieval techniques to abstract requirements from existing specifications, typically expressed in natural language. 

%These works do not consider precise logical dependencies and solely focus on ontological semantics. As a result users have to manually set the variability information. Moreover a risk is to build an FM in contradiction with the actual dependencies of a system.

Another related subject is constraint mining. In~\cite{acher-cleve-etal:2013}, architectural and expert knowledge as well as plugins dependencies are combined to obtain an exploitable and maintainable feature model. Nadi \textit{et al.}~\cite{nadi2014} developed a comprehensive infrastructure to automatically extract configuration constraints from C code. 
%Their empirical study showed that many of the constraints require expert knowledge or more specific analysis. Our experience in a very different context concurs with the findings. A substantial amount of constraints cannot be inferred only from the analysis of languages components, despite the development of specific heuristics for mining constraints.

All these works present innovative techniques for mining and extracting FMs from various artefacts. However, they do not support the synthesis of attributes despite the presence of non Boolean data in some of these artefacts.
In this report, we do not consider such a broad view; we focus solely on the synthesis of AFMs. 
Automated extraction techniques can be used to process different types of artefacts and eventually fed our synthesis algorithm. 
\textit{The foundations presented in this report open avenues for investigating novel reverse engineering scenarios involving attributes.}

\subsection{Language and Tool Support for Feature Models}
There are numerous existing academic (or industrial) languages and tools for specifying and reasoning about FMs~\cite{benavides2010,berger2013}.
 
\emph{FeatureIDE}~\cite{thm2012featureide,batory2009} is an Eclipse-based IDE that supports all phases of feature-oriented software development. 
\emph{SPLConqueror} is a tool to measure and optimize non-functional properties in software product lines~\cite{SPLConqueror,siegmundKKABRS12,siegmund2013}. 
Descriptions of product lines include feature attributes such as pricing, footprint etc.
 
\emph{FAMA (Feature Model Analyser)}~\cite{benavides2010} is a framework for the automated analysis of FMs integrating some of the most commonly used logic representations and solvers. It supports attributes with integer, real and string domains. % \gb{check supported types of attributes for FAMA}
SPLOT~\cite{SPLOTlong} provides a Web-based environment for editing and configuring FMs.
 S2T2~\cite{S2T2} is a tool for the configuration of large FMs.
Commercial solutions (\emph{pure::variants}~\cite{purevariants} and \emph{Gears}~\cite{krueger2007biglever}) also provide a comprehensive support for product lines (from FMs to model/source code derivation). 
\emph{TVL}~\cite{TVL} is a language supporting several extensions to FMs such as attributes.
\emph{Clafer}~\cite{bak2010} is a framework mixing feature models and meta-models that supports the definition and analysis of attributes.
 Seidl \etal~\cite{seidl2014capturing} propose an extension of FMs for supporting variability in time and space. The so-called hyper FMs can be considered as a special case of AFMs in which attributes' domains are graphs of version. Alferez \etal reported an experience in the video domain involving numerous numerical values and meta-information, encoded as attributes with the VM language~\cite{alferez:hal-01023159}.

\textit{None of the existing tools propose support for synthesizing attributed FMs.}

\section{\label{sec:motivation}Background and Motivation}
This report aims to describe the foundations for synthesizing an attributed feature model from product descriptions\footnote{Roughly speaking, product descriptions are represented as a matrix, each line documenting a product along different Boolean or numerical values. More details will be given in the remainder of the report.}. %(see more details hereafter).
In this section, we describe background information related to attributed feature models. 
We then explain why attributes are an essential extension to feature models in order to support the expressiveness of product descriptions.
We further motivate the need for an automated encoding of product descriptions.

\subsection{Attributed Feature Models} % Background}
\label{sec:background}
%\subsection{Attributed Feature Models}
%Feature modeling consists in defining the set of valid configurations of a system. 
%The diagrammatic representation of an FM completes this configuration semantics by organizing the features and describing their relations.
%Extending FMs with attributes...
Several formalisms supporting attributes exist~\cite{benavides2007fama,TVL,bak2010,czarnecki2004}.
In this report, we use a formalism inspired from FAMA~\cite{benavides2007fama}.
An AFM is composed of a feature diagram (see Definition~\ref{def:afd}) and an arbitrary constraint (see Definition ~\ref{def:afm}).

\begin{mydef}[Attributed Feature Diagram]
\label{def:afd}
An attributed feature diagram FD is a tuple $\langle F$, $H$, $E_{M}$, $G_{MTX}$, $G_{XOR}$, $G_{OR}$, $A$, $D$, $\delta$, $\alpha$, $RC \rangle$ such that:
\begin{itemize}
\item $F$ is a finite set of boolean features.
\item $H = (F, E)$ is a rooted tree of features where $E \subseteq F \times F$ is a set of directed child-parent edges.
\item $E_{M} \subseteq E$ is a set of edges that define mandatory features.
\item $G_{MTX}, G_{XOR}, G_{OR} \subseteq P(E\backslash{}E_m)$ are sets of feature groups. Each feature group is a set of edges. The feature groups of $G_{MTX}$, $G_{XOR}$ and $G_{OR}$ are non-overlapping and all edges in a group share the same parent.
\item $A$ is a finite set of attributes.
%\item $D$ is a set of possible domains for the attributes in $A$. A domain is a tuple $\langle V,0_D,\leq \rangle$ with $V$ a potentially infinite set of values, $0_D \in V$ the value that an attribute takes when it is not selected and $\leq \in V \times V$ a partial order.
\item $D$ is a set of possible domains for the attributes in $A$. 
\item $\delta \in A \rightarrow D$ is a total function that assigns a domain to an attribute. 
\item $\alpha \in A \rightarrow F$ is a total function that assigns an attribute to a feature. 
\item $RC$ is a set of constraints over $F$ and $A$ that are considered as human readable and may appear in the feature diagram in a graphical or textual representation (\eg binary implication constraints can be represented as an arrow between two features).
\end{itemize}

A domain $d \in D$ is a tuple $\langle V_d, 0_d, <_d \rangle$ with $V_d$ a finite set of values, $0_d \in V_d$ the null value of the domain and $<_d$ a partial order on $V_d$.
When a feature is not selected, all its attributes bound by $\alpha$ take their null value, \ie $\forall (a,f) \in \alpha \text{ with } \delta(a) = \langle V_a, 0_a, <_a \rangle \text{, we have } \neg f \Rightarrow (a = 0_a)$.

\end{mydef}

\begin{figure}[h]
\centering
%\scalebox{0.9}{
\begin{small}
\framebox[\columnwidth]{%\minibox{
\begin{tabular}{l c l}
\textit{readable\_constraint}			& ::= & 		\textit{bool\_factor} `\textbf{\textit{$\Rightarrow$}}' \textit{bool\_factor} ;\\
%\textit{bool\_factor}			& ::= & 		{\bf feature\_name} $|$ \textit{rel\_expr} $|$ \\
%&&													\textit{inclusion\_stment}  $|$ `{\bf \textit{not}}' \textit{bool\_factor} ;\\
\textit{bool\_factor}			& ::= & 		{\bf feature\_name}  $|$ `{\bf \textit{$\neg$}}' \textbf{feature\_name}  $|$ \textit{rel\_expr} ;\\
\textit{rel\_expr}				& ::= & 		{\bf attribute\_name}  \textit{rel\_op} {\bf num\_literal} ;\\
\textit{rel\_op}					& ::= & 		`$>$' $|$ `$<$' $|$ `$\geq$' $|$ `{\bf $\leq$}' $|$ `{\bf =}' ;\\
%\textit{inclusion\_stment}	& ::= & 		{\bf attribute\_name}  `\textbf{\textit{in}}' \textit{domain\_expr} ;\\
%\textit{domain\_expr}			& ::= & 		`[' {\bf num\_literal} `. .' {\bf num\_literal} `]' $|$\\
%&&													`\{'{\bf num\_literal} (`,' {\bf num\_literal})*`\}' ;\\
\end{tabular}
%}
}
\end{small}
%}
\caption{\label{fig_rcGrammar} The grammar of readable constraints.}
\end{figure}

For the set of constraints in $RC$, formally defining what is human readable is essential for automated techniques.
In this report, we define $RC$ as the constraints that are consistent with the grammar in Figure~\ref{fig_rcGrammar}. Some examples of such constraints can be found in the bottom of Figure~\ref{fig:afm-b}.
We consider that these constraints are small enough and simple enough to be human readable.
In this grammar, each constraint is a binary implication, which specifies a relation between the values of two attributes or features.
Feature names and relational expressions over attributes are the boolean factors that can appear in an implication. 
Further, we only allow natural numbers as numerical literals (num\_literal).

%Figure~\ref{fig_rcGrammar} is a grammar for $RC$, the readable constraints in Definition 1. Each readable constraint is a binary implication, which specifies a relation between the values of two attributes or features. Feature names, relational expressions, and inclusion statements are the boolean factors that can appear in an implication. The grammar does not allow expressing a relation between two attributes. Further, we only allow natural numbers as numerical literals (num\_literal).

The grammar of Figure~\ref{fig_rcGrammar} and the formalism of attributed feature diagrams (see Definition~\ref{def:afd}) are not expressive enough to represent any possible configuration matrix~\cite{schobbens2007}.
Therefore, to exactly represent a configuration matrix, we need to add an arbitraty constraint to the feature diagram. This leads to the definition of an attributed feature model (see Definition~\ref{def:afm}).

\begin{mydef}[Attributed Feature Model]
\label{def:afm}
A feature model is a pair $\langle FD, \Phi \rangle$ where FD is an attributed feature diagram and $\Phi$ is an arbitrary constraint over $F$ and $A$ that represent the constraints that cannot be expressed by $RC$.
\end{mydef}

\myexample{
Figure~\ref{fig:afm-b} shows an example of an AFM describing a product line of Wiki engines. 
The feature \feature{WikiMatrix} is the root of the hierarchy. It is decomposed in 3 features: \feature{LicenseType} which is mandatory and \feature{WYSIWYG} and \feature{LanguageSupport} which are optional.
The xor-group composed of \feature{GPL}, \feature{Commercial} and \feature{NoLimit} defines that the wiki engine has exactly 1 license and it must be selected among these 3 features.
The attribute \feature{LicensePrice} is attached to the feature \feature{LicenseType}. The attribute's domain states that it can take a value in the following set: \{0, 10, 20\}.
The readable constraints and $\Phi$ for this AFM are listed below its hierarchy (see Figure~\ref{fig:afm-b}). The first one restricts the price of the license to 10 when the feature \feature{Commercial} is selected.
}

As illustrated by the example, an AFM has two main objectives. First, it defines the valid configurations of a product line. This corresponds to the \emph{configuration semantics} of the AFM (see Definition~\ref{def:configuration_semantics}). Second, it organizes the relationships between the features and attributes which can be of different types (\eg decomposition or specialization). This corresponds to the \emph{ontological semantics} of the AFM (see Definition~\ref{def:ontological_semantics}).

\begin{mydef}[Configuration semantics]
\label{def:configuration_semantics}
A configuration of an AFM $m$ is defined as a set of selected features and a value for every attribute. The configuration semantics $\llbracket m \rrbracket$ of $m$ is the set of valid configurations.
\end{mydef}

\begin{mydef}[Ontological semantics]
\label{def:ontological_semantics}
The hierarchy $H$, the feature groups ($G_{MTX}$, $G_{OR}$ and $G_{XOR}$), the place of the attributes defined by $\alpha$ and the constraints $RC$ form the ontological semantics of an attributed feature model. It represents the semantics of features and attributes' relationships including their structural relationships and conceptual proximity.
%\gb{Do we need ontological semantics for this paper?}
\end{mydef}

%\subsection{Constraint Propagation}
% RB: constraint propagation instead of constraint programming, which is more generic.

% 

\subsection{Product Descriptions and Feature Models}
Product descriptions are usually represented in tabular format, such as spreadsheets and product comparison matrices. The objective of such formats is to describe the characteristics of a set of products in order to document and differentiate them. From now on, we will use the term \textit{configuration matrix} to refer to these tabular formats (see Definition~\ref{def:matrix} for a formal description). 

For instance, consider the domain of Wiki engines that we will use as a running example throughout the report. The list of features supported by a set of Wiki engines can be documented using a configuration matrix. Figure~\ref{fig:afm-a} is a very simplified configuration matrix, which provides information about eight different Wiki engines. 

\begin{mydef}[Configuration matrix]
\label{def:matrix}
Let $\mathbf{c}_1, ..., \mathbf{c}_M$ be a given set of configurations. Each configuration $\mathbf{c}_i$ is an $N$-tuple $(c_{i,1}, ..., c_{i,N})$, where each element $c_{i,j}$ is the value of a \emph{variable} $V_j$. A variable represents either a feature or an attribute. Using these configurations, we create an $M \times N$ matrix $\mathbf{C}$ such that $\mathbf{C} = [\mathbf{c}_1, ..., \mathbf{c}_M]^t$, and call it a \emph{configuration matrix}.
\end{mydef}

\begin{figure}%
\centering
\subfigure[A configuration matrix for Wiki engines.]{
\label{fig:afm-a}
\small
\begin{tabular}{|c|c|c|c|c|c|}
\hline 
\multirow{2}{*}{\bf{Identifier}} & \multirow{2}{*}{\bf{LicenseType}} & \multirow{2}{*}{\bf{LicensePrice}} & \bf{Language} & \multirow{2}{*}{\bf{Language}} & \multirow{2}{*}{\bf{WYSIWYG}} \\ 
& & & \bf{Support} & & \\
\hline 
\bf{Confluence} & Commercial & 10 & Yes & Java & Yes \\ 
\hline
\bf{PBwiki} & NoLimit & 20 & No & -- & Yes \\ 
\hline 
\bf{SimpleWiki} & NoLimit & 10 & No & -- & Yes \\ 
\hline 
\bf{MoinMoin} & GPL & 0 & Yes & Python & Yes \\ 
\hline 
\bf{TWiki} & GPL & 0 & Yes & Perl & Yes \\ 
\hline
\bf{PerlWiki} & GPL & 10 & Yes & Perl & Yes \\ 
\hline   
\bf{MediaWiki} & GPL & 0 & Yes & PHP & No \\ 
\hline  
\bf{PHPWiki} & GPL & 10 & Yes & PHP & Yes \\ 
\hline   
\end{tabular}
}
\subfigure[One possible attributed feature model for representing the configuration matrix in Figure~\ref{fig:afm-a}]{
\label{fig:afm-b} 
\includegraphics[width=.85\textwidth]{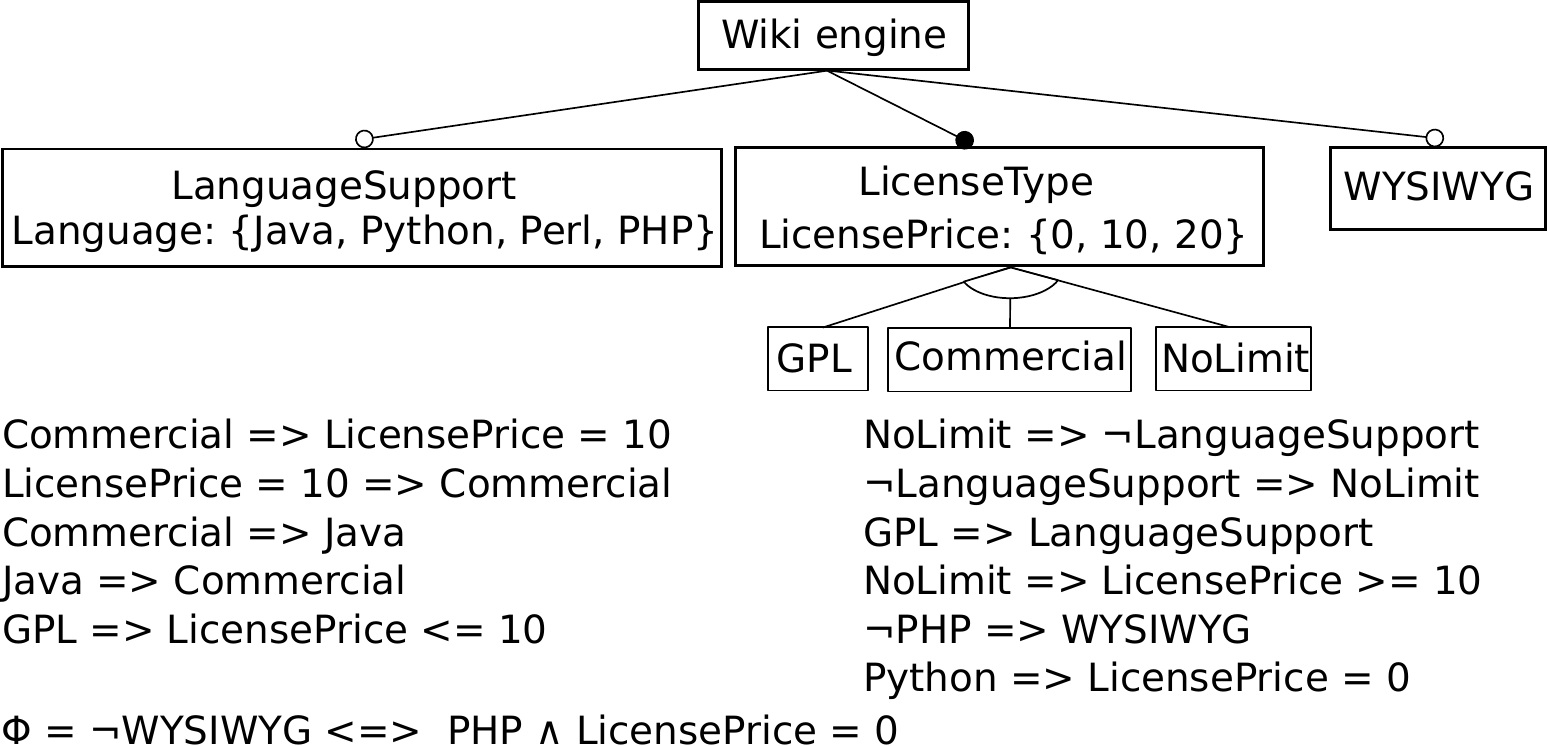}
}
\subfigure[Domain knowledge for synthesizing Figure~\ref{fig:afm-b} from Figure~\ref{fig:afm-a}.]{
\label{fig:afm-c}
\small
\begin{tabular}{|l|p{10cm}|}
\hline 
\bf{Information} & \bf{Value} \\ 
\hline 
Features & \feature{WYSIWYG}, \feature{LanguageSupport}, \feature{LicenseType}, \feature{GPL}, \feature{Commercial}, \feature{NoLimit} \\ 
\hline
Interpretation of cells & "Yes" = presence of a feature, "No" = absence of a feature\\
\hline
Root & \feature{Wiki engine} \\ 
\hline
Hierarchy (child $ \rightarrow $ parent) & 
	\feature{LanguageSupport} $ \rightarrow $ \feature{Wiki engine},
	\feature{LicenseType} $ \rightarrow $ \feature{Wiki engine},
	\feature{WYSIWYG} $ \rightarrow $ \feature{Wiki engine},
 	\feature{GPL} $ \rightarrow $ \feature{LicenseType},
 	\feature{Commercial} $ \rightarrow $ \feature{LicenseType},
	\feature{NoLimit} $ \rightarrow $ \feature{LicenseType} \\
\hline
Attributes & \feature{Language}, \feature{LicensePrice}\\ 
\hline
Domains & "--" is the null value of \feature{Language} \\
\hline
Place of attributes & $\alpha(\feature{Language}) = \feature{LanguageSupport}$, $\alpha(\feature{LicensePrice}) = \feature{LicenseType}$\\
\hline
Feature groups & $\{$\feature{GPL}, \feature{Commercial}, \feature{NoLimit}$\}$\\
\hline
Interesting values for RC & (\feature{LicensePrice}, 10)\\
\hline
\end{tabular}
}
\caption{A motivating example}
\label{fig:afm}
\end{figure}

% Configuration matrices abound on the internet and are a common practice to model variability in industry~\cite{ASE2013PCMs,berger2013,sannier:hal-00927312,becan2014automating}. 
Conceptually, a configuration matrix documents a family of related products (e.g., an SPL). % and provides an \emph{intensional} definition.   
Feature models can also be used to document a set of configurations (products).
Though feature models and configuration matrices share the same goal, their syntax and usage vastly differ. 

In contrast to configuration matrices, feature models do not provide an explicit listing (or \emph{enumerative} definition) of the set of products. 
% provide an \emph{extensional} definition of a set of product. 
% , which is only possible for finite sets and only practical for relatively small sets, is a type of enumerative definition.
 From this perspective, feature models are rather a compact representation of a set of products; variability information (e.g., optionality) and cross-tree constraints define the legal configurations corresponding to products. Domain analysts can quickly visualize what are the relationships between features. First, the variability information is made explicit and can be directly read and understood. Moreover the hierarchy helps to structure the information and a potentially large number of features into multiple levels of increasing detail~\cite{czarnecki2006}. 
 Besides, an enumerative definition of a set of products is only practical for relatively small sets; feature models are particularly suited when the enumeration of the set is impractical. 
 
 Configuration matrices and feature models are semantically related and aim to characterize a set of configurations. 
 The two formalisms are complementary; one would like to switch from one representation to the other.

% Feature models captur

% However, the lack of formalization, tools and scalability hinder the manual exploitation of configuration matrices~\cite{becan2014automating}.
% To target this challenge, we intend to derive feature models from configuration matrices. 
%  enable the use and development of reasoning techniques based on such models.

\subsection{Synthesis of Attributed Feature Models}

%\gb{Feature models are good for... They provide a formalism for describing product lines.}
%Our objective in this paper is to build a \emph{model} that provides a succinct representation of the commonalities and variabilities of the products in such product descriptions. This representation can then be used for various types of reasoning, such as BLAH and BLAH.
Our first goal is to further the understanding of the relation between the two formalisms (\textbf{RQ1}). 
 Our second goal is to provide synthesis mechanisms to transition from configuration matrices to feature models (\textbf{RQ2}).  
 
More specifically, we want to formalize the relationship between configuration matrices and \emph{attributed} feature models (\textbf{RQ1}). 
Previous works~\cite{AcherVAMOS12,andersen2012} limit their study to Boolean constructs. However, non Boolean data (\eg numbers, strings or dates) are intensively employed in configuration matrices to document variability~\cite{ASE2013PCMs,sannier:hal-00927312,becan2014automating}. 
For instance, the price of the license is represented by an integer in Figure~\ref{fig:afm-a}.
 These kinds of data are typically encoded as attributes in a feature model. 
 Figure~\ref{fig:afm-b} shows an attributed feature model, and a set of constraints that together provide a possible representation of the comparison matrix in Figure~\ref{fig:afm-a}. In the attributed feature model, the feature \feature{LicenseType} contains an attribute named \feature{LicensePrice} that represents the price of the license. 
In general, a configuration matrix can be represented by a multiplicity of feature models. To choose a unique one among them, we use some extra information that we call \emph{domain knowledge}. The domain knowledge that is used for generating the feature model in Figure~\ref{fig:afm-b} is shown in Figure~\ref{fig:afm-c}.
 
With regards to \textbf{RQ2}, numerous works reported that the manual development of a feature model is time-consuming and error-prone~\cite{FSE2013,AcherVAMOS12,acher-cleve-etal:2013,she2011,ryssel2011,bagheri2012Wikipedia,FerrariSd13,alves2008,niu2009,rashid2009,chen2005,nadi2014,czarnecki2007,she2011,janota2008,AcherVaMoS13,lopez2011,Lopez-HerrejonGBSE12,herejonFASE2013,andersen2012,she2008,she2014efficient,AcherVaMoS13,becanESE2015,LopezHerrejon2014}. 
% Automated synthesis techniques can be used to tackle this problem.
% However,  
We assist users in synthesizing a consistent and meaningful feature model -- this time with attributes. 

%  but no support exist for the synthesis of attributed feature models (see Section~\ref{sec:related}). 
% Our objective is to develop techniques for supporting the synthesis of attributed feature models.

% \gb{detail example}

% Such techniques would open avenues for reverse engineering new kinds of artefacts that contain variability such as configuration matrices. New data sets of attributed feature models could be created to challenge existing and future techniques of the variability community.
% It would also enable the use of existing techniques operating over quantitative information amenable to attributes. For instance, 
% \gb{copied from related work \\
% predictive models operate over quantitative information amenable to (numerical) attributes, kolesnikovASSKS13}

% model checking~\cite{cordy2013}

% ... exact and approximate methods for multi-objective~\cite{DBLP:conf/kbse/GuoZORCAA14,guo2014}

% stats
%}

\subsection{Synthesis Scenarios}

\begin{figure}
\centering
\includegraphics[scale=0.45]{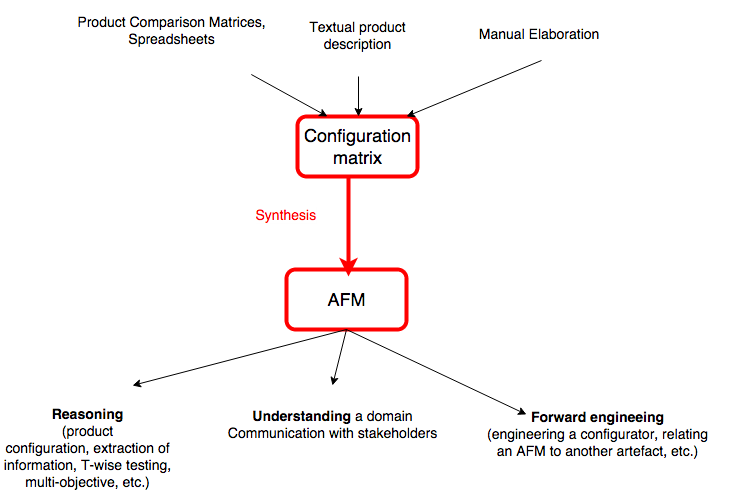}
\caption{\label{fig:motivation}Core problem: synthesis of attributed feature model from configuration matrix}
\end{figure}

Figure~\ref{fig:motivation} summarizes our objective. 
On the one hand, many kinds of artefacts or problems are amenable to the formalism of configuration matrix (see on top).  
Product comparison matrices (PCMs)~\cite{ASE2013PCMs,sannier:hal-00927312,becan2014automating} abound on the internet and are describing products along different criteria. Modulo a fixed and precise interpretation of cell values, PCMs can be encoded as configuration matrices. 
Our recent initiative for automating the formalization and providing state-of-the-art, specialized editors for PCMs can help for this purpose~\cite{becan2014automating}.
Some textual descriptions of product can also be encoded as a configuration matrix~\cite{FSE2013}. 
A manual listing of configurations can also be created and maintained. Berger \etal~\cite{berger2014} reported that practitioners use FMs to manage a set of configurations. Guo \etal compute similar configuration matrices for assessing non-functional properties of products with the goal of building predictive models~\cite{guo2013}.
The concrete scenario that emerges is as follows: automated techniques can encode some artefacts as configuration matrix and eventually fed a synthesis algorithm.  
% reported that 

On the other hand, the synthesis of AFMs has three main motivations:
\begin{itemize}
\item \textbf{reasoning}: efficient techniques, relying on either CSP solvers, BDD, SAT, or SMT solvers, have been developed to automatically compute relevant information~\cite{benavides2010}, model-check a product line~\cite{classen2010,cordy2013,thuem2014,beek2013}, automate product configuration~\cite{ebrahim2012,ebrahim2013,hubaux2013Book}, resolve multi-objective problems~\cite{DBLP:conf/kbse/GuoZORCAA14}, or compute T-wise configurations~\cite{GalindoISSTA14,DBLP:conf/splc/JohansenHF12,apel2013,DBLP:journals/tse/HenardPPKHT14,thuem2014}. The encoding of a configuration matrix as an AFM provides the ability to reuse state-of-the-art reasoning techniques;
\item \textbf{communication and understanding}: as any model, an AFM can be used to communicate with other stakeholders~\cite{berger2014} inside or outside a given organization. Domain analysts and product managers can also understand a given domain, market, or family of products; 
\item \textbf{forward engineering}: an AFM is central to many product line approaches and can serve as basis for a forward engineering. For instance, the engineering of a configurator~\cite{ebrahim2012,ebrahim2013} can be envisioned. The hierarchy, the explicit list of features as well as the presence of variability information make the derivation of a user interface quite immediate. Figure~\ref{fig:configurator} depicts a possible configurator that could be engineered from the AFM of Figure~\ref{fig:afm-b}. Besides an AFM can be related to other artefacts (e.g., source code) for automating the derivation of products.
\end{itemize}

% Other scenarios are possible

%  the lack of formalization, tools and scalability hinder the manual exploitation of configuration matrices

\begin{figure}
\centering
\includegraphics[scale=0.25]{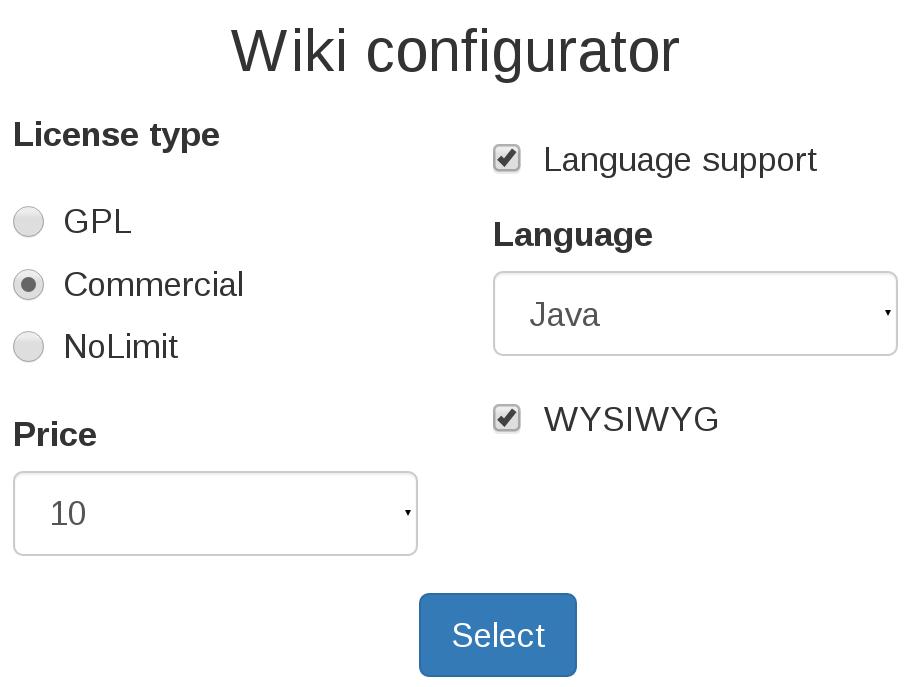}
\caption{A possible configurator generated from the AFM of Figure~\ref{fig:afm-b}}
\label{fig:configurator}
\end{figure}

%\section{Background}
%\label{sec:background}
%\input{background.tex}

\section{Synthesis Problem} % of Attributed Feature Models
\label{sec:problem}
The problem tackled in this report is to synthesize an AFM (see Definition~\ref{def:afm}) from a configuration matrix (see Definition~\ref{def:matrix}). Two main challenges arise: first, to preserve the configuration semantics of the input matrix; second, to produce a maximal and readable diagram.
%
%Let $\mathbf{c}_1, ..., \mathbf{c}_M$ be a given set of configurations. Each configuration $\mathbf{c}_i$ is an $N$-tuple $(c_{i,1}, ..., c_{i,N})$, where each element $c_{i,j}$ is the value of a \emph{variable} $V_j$. Using these configurations, we create an $M \times N$ matrix $\mathbf{C}$ such that $\mathbf{C} = [\mathbf{c}_1, ..., \mathbf{c}_M]^t$, and call it a \emph{configuration matrix}.
%
%From this configuration matrix, the objective is to synthesize an AFM that has the same configuration semantics as the input matrix.

Synthesizing an AFM that represents the exact same set of configurations (\ie configuration semantics) as the input configuration matrix is primordial. 
If the AFM is too permissive, it would expose the user to illegal configurations.
To prevent this situation, the algorithm must be sound (see Definition~\ref{def:soundness}).
Conversely, if the AFM is too constrained, it would prevent the user from selecting available configurations, resulting in unused variability.
Therefore, the algorithm must also be complete (see Definition~\ref{def:completeness}).
Figure~\ref{fig:soundness_completeness} illustrates how these two properties are related to the configuration semantics of the input configuration matrix.

\begin{mydef}[Soundness of AFM Synthesis]
\label{def:soundness}
A synthesis algorithm is sound if the resulting AFM ($afm$) represents only configurations that exist in the input configuration matrix ($cm$), \ie
%An AFM ($afm$) is sound w.r.t a configuration matrix ($cm$) if the $afm$ represents only configurations that exist in $cm$, \ie
$\llbracket afm \rrbracket \subseteq \llbracket cm \rrbracket$.
\end{mydef}

\begin{mydef}[Completeness of AFM Synthesis]
\label{def:completeness}
A synthesis algorithm is complete if the resulting AFM ($afm$) represents at least all the configurations of the input configuration matrix ($cm$), \ie
%An AFM ($afm$) is complete w.r.t a configuration matrix ($cm$) if the $afm$ represents at least all the configurations of $cm$, \ie
$\llbracket cm \rrbracket \subseteq \llbracket afm \rrbracket$.
\end{mydef}

\begin{figure}
\centering
\includegraphics[width=0.5\textwidth]{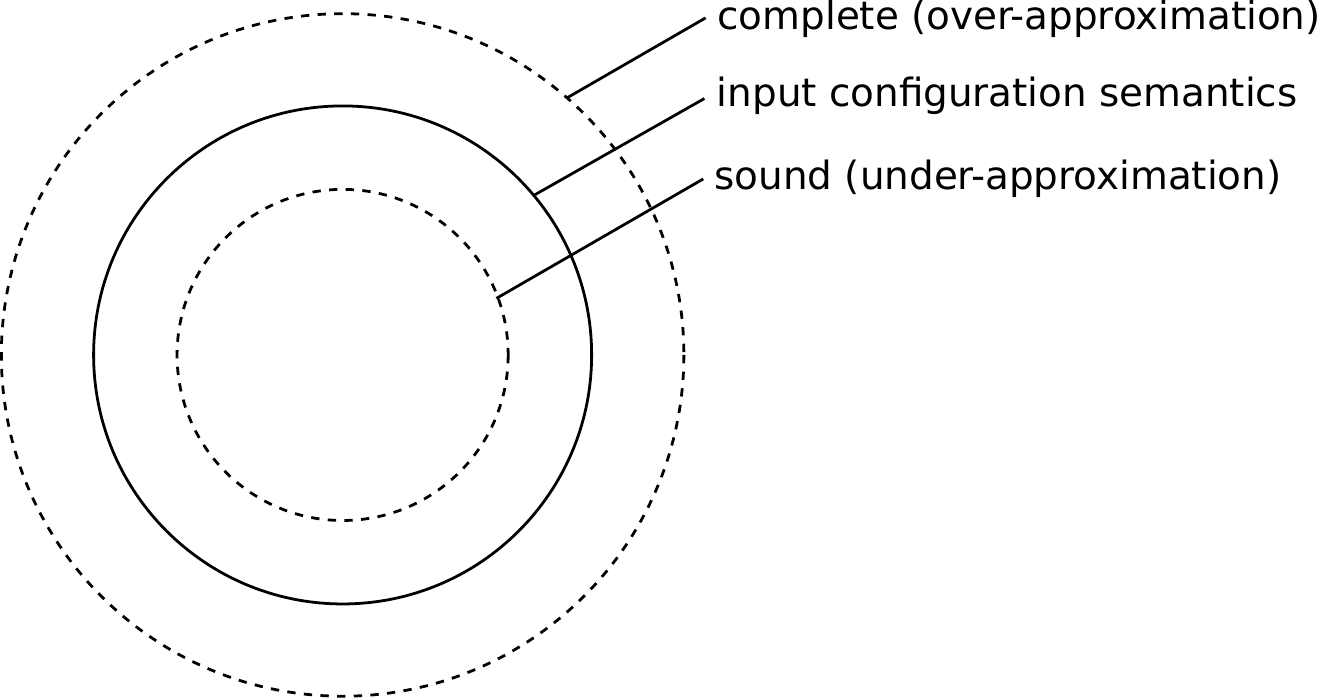}
\caption{\label{fig:soundness_completeness}Soundness and completeness w.r.t the input configuration semantics}
\end{figure}

To avoid the synthesis of a trivial AFM (\eg an AFM with the input matrix encoded in the constraint $\Phi$ and no hierarchy, \ie $E=\emptyset$), we target a maximal AFM as output (see Definition~\ref{def:maximal_afm}). Intuitively, we enforce that the feature diagram contains as much information as possible. 
%In particular, we note that the grammar of $RC$ can be further restricted by the domain knowledge to avoid uninteresting constraints. 
%\rb{I feel like the previous sentence does not belong to here. It is sufficient that the restriction on RC is mentioned in the definition of the domain knowledge. We can still use it in the proof of maximality.}
%In that case, the domain knowledge defines a set of values that are allowed as numerical literals.
Definition~\ref{def:problem} defines the AFM synthesis problem targeted in this report.

\begin{mydef}[Maximal Attributed Feature Model]
\label{def:maximal_afm}
An AFM is maximal if its hierarchy $H$ connects every feature in $F$ and if none of the following operations are possible without modifying the configuration semantics of the AFM:
\begin{itemize}
\item add an edge to $E_{M}$
\item add a group to $G_{MTX}, G_{XOR}$ or $G_{OR}$
\item move a group from $G_{MTX}$ or $G_{OR}$ to $G_{XOR}$
\item add to $RC$ a non-redundant constraint that conforms to the restrictions specified in the domain knowledge.
%\item add a non-redundant constraint to $RC$ in which the "num\_literal" (see grammar of $RC$ in Figure~\ref{fig_rcGrammar}) is contained in the interesting values that are defined by the domain knowledge.
\end{itemize}
\end{mydef}

\begin{mydef}[Attributed Feature Model Synthesis Problem]
\label{def:problem}
Given a set of configurations $\text{sc}$, the problem is to synthesize an AFM $m$ such that $\llbracket sc \rrbracket = \llbracket m \rrbracket$ (\ie the synthesis is sound and complete) and $m$ is maximal.
%that is sound, complete and exhibiting an appropriate hierarchy.
\end{mydef}

\subsection{Equivalence of Attributed Feature Models}
\label{subsec:equivalence}
In Definition~\ref{def:problem}, we enforce the AFM to be maximal to avoid trivial solutions to the synthesis problem. Despite this restriction, the solution to the problem may not be unique (see Definition~\ref{def:afm_diff}).
Given a set of configurations (\ie a configuration matrix), multiple maximal AFMs can be synthesized.

\begin{mydef}[Equivalence of Attributed Feature Models]
\label{def:afm_diff}
Two AFMs ($afm_1$ and $afm_2$) are equivalent if $ \llbracket afm_1 \rrbracket = \llbracket afm_2  \rrbracket$ and if their feature diagrams (see Definition~\ref{def:afd}) are equal (\ie only $\Phi$ can vary between the two AFMs).
%\land H_1 = H_2 \land E_{m1} = E_{m2} \land G_{MTX1} = G_{MTX2} \land G_{OR1} = G_{OR2} \land G_{XOR1} = G_{XOR2} \land A_1 = A_2 \land D_1 = D_2 \land \delta_1 = \delta_2 \land \alpha_1 = \alpha_2$.

\end{mydef}

This property has already been observed for the synthesis of boolean FMs~\cite{she2014efficient,AcherVaMoS13,czarnecki2007,becanESE2015}. 
Several hierarchies of features can exist for the same configuration semantics. 
Extending boolean FMs with attributes exacerbates the situation. In some cases, the place of the attributes and the constraints over them can be modified without affecting the configuration semantics of the synthesized AFM.

\myexample{%
Figure~\ref{fig:afm-b} and Figure~\ref{fig:afm2} depict two AFMs representing the same configuration matrix of Figure~\ref{fig:afm-a}. They have the same configuration semantics but their attributed feature diagrams are different. In Figure~\ref{fig:afm-b}, the feature \feature{WYSIWYG} is placed under \feature{Wiki engine} while in Figure~\ref{fig:afm2}, it is placed under the feature \feature{LicenseType}. The attribute \feature{LicensePrice} is placed in feature \feature{LicenseType} in Figure~\ref{fig:afm-b}, while it is placed in feature \feature{Wiki engine} for the AFM in Figure~\ref{fig:afm2}.
}

\begin{figure}
\centering

\subfigure[\label{fig:afm2}Another attributed feature model representing the configuration matrix in Figure~\ref{fig:afm-a}]{
\includegraphics[width=0.85\textwidth]{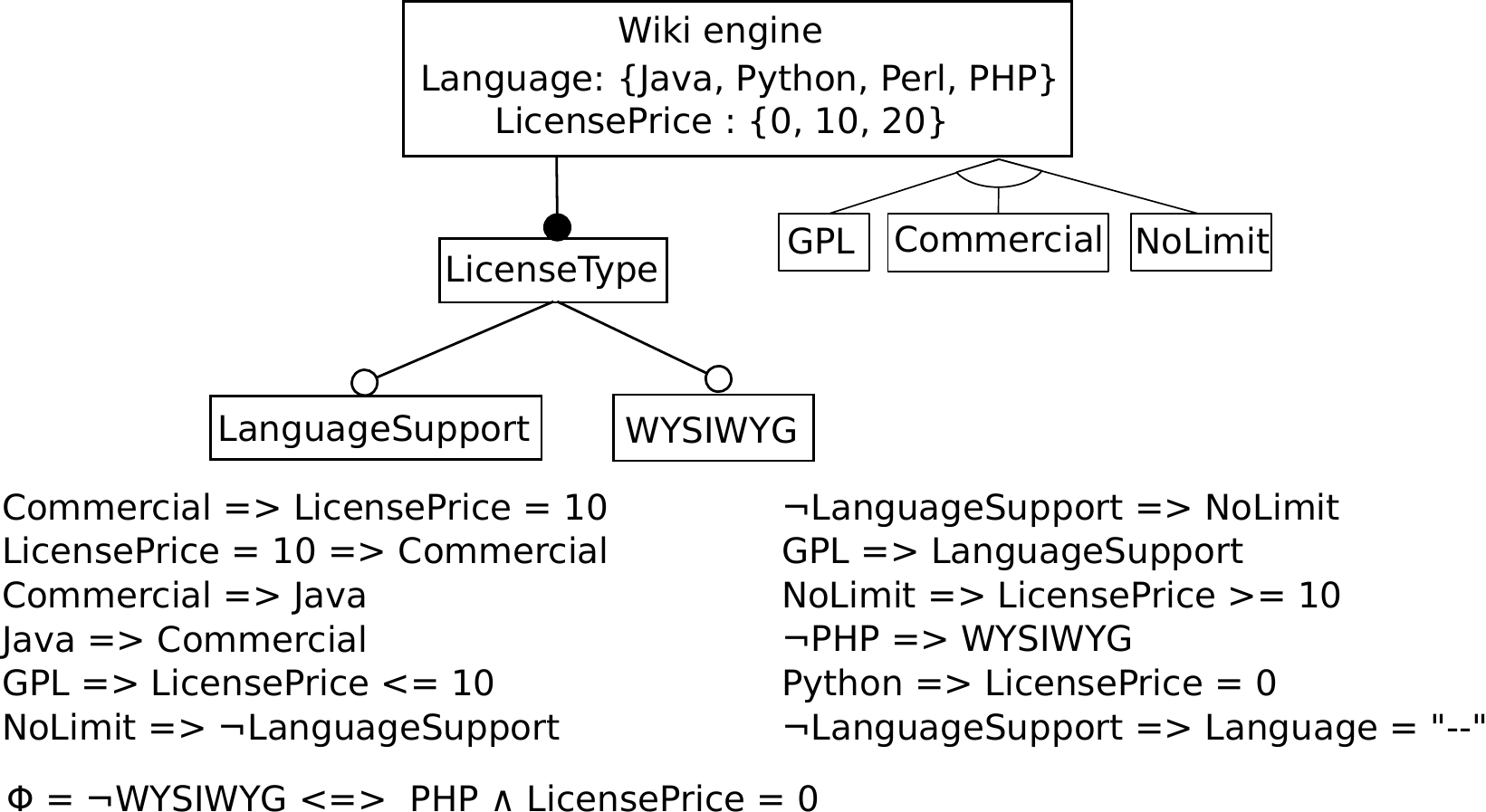}
}

\subfigure[Domain knowledge for synthesizing Figure~\ref{fig:afm2} from Figure~\ref{fig:afm-a}.]{
\label{fig:afm2-dk}
\small
\begin{tabular}{|l|p{10cm}|}
\hline 
\bf{Information} & \bf{Value} \\ 
\hline 
Features & \feature{WYSIWYG}, \feature{LanguageSupport}, \feature{LicenseType}, \feature{GPL}, \feature{Commercial}, \feature{NoLimit} \\ 
\hline
Interpretation of cells & "Yes" = presence of a feature, "No" = absence of a feature\\
\hline
Root & \feature{Wiki engine} \\ 
\hline
Hierarchy (child $ \rightarrow $ parent) & 
	\feature{LanguageSupport} $ \rightarrow $ \feature{LicenseType},
	\feature{LicenseType} $ \rightarrow $ \feature{Wiki engine},
	\feature{WYSIWYG} $ \rightarrow $ \feature{LicenseType},
 	\feature{GPL} $ \rightarrow $ \feature{Wiki engine},
 	\feature{Commercial} $ \rightarrow $ \feature{Wiki engine},
	\feature{NoLimit} $ \rightarrow $ \feature{Wiki engine} \\
\hline
Attributes & \feature{Language}, \feature{LicensePrice}\\ 
\hline
Domains & "--" is the null value of \feature{Language} \\
\hline
Place of attributes & $\alpha(\feature{Language}) = \feature{Wiki engine}$, $\alpha(\feature{LicensePrice}) = \feature{Wiki engine}$\\
\hline
Feature groups & $\{$\feature{GPL}, \feature{Commercial}, \feature{NoLimit}$\}$\\
\hline
Interesting values for RC & (\feature{LicensePrice}, 10)\\
\hline
\end{tabular}
}
\caption{Another possible attributed feature model for the motivating example.}
\end{figure}

\subsection{Synthesis Parametrization}

As shown previously, several AFMs with the same configuration semantics can be synthesized from a single matrix. To synthesize a unique AFM (see Definition~\ref{def:afm_diff}), our algorithm must take decisions.
These decisions are based on what we call the \emph{domain knowledge} which can come from heuristics, ontologies or a user of our algorithm. This domain knowledge can be provided interactively during the synthesis or as input before the synthesis.
By providing this knowledge, we parametrize the algorithm to obtain a unique AFM.

The domain knowledge can be represented as a set of functions that perform the following operations:
\begin{itemize}
%\item the representation of the columns and cells of the matrix (features, attributes, domains)
\item decide if a column should be represented as a feature or an attribute
%\item the interpretation of the cells
\item give the interpretation of the cells (type of the data, partial order)
%\item the hierarchy of the features
\item select a possible hierarchy
%\item the place of the attributes
\item select a place for each attribute among their legal positions
%\item the feature groups
\item select a feature group among the overlapping ones.
%\item the shape of the constraints in RC
\item provide interesting bounds for each attribute in order to compute meaningful constraints for $RC$
\end{itemize}

Examples of the information contained in domain knowledge can be found in Figures~\ref{fig:afm-c} and~\ref{fig:afm2-dk}.

\subsection{Over-approximation of the Attributed Feature Diagram}
\label{sec:overapprox}
A crucial property of the output AFM of the synthesis problem is to have the exact same configuration semantics as the input configuration matrix (see Definition~\ref{def:problem}).
As shown in Section~\ref{sec:background}, the attributed feature diagram may over-approximate the configuration semantics, i.e.,
$\llbracket cm \rrbracket \subseteq \llbracket FD \rrbracket$.
Therefore the additional constraint $\Phi$ of an AFM (see Definition~\ref{def:afm}) is required for providing an accurate representation for any arbitrary configuration matrix.

A basic strategy for computing $\Phi$ is to directly encode the configuration matrix as a constraint, i.e., $\llbracket cm \rrbracket =\ \llbracket \Phi \rrbracket$ (see Equation~\ref{eq:phi}). An advantage is that the computation is immediate and $\Phi$ is, by construction, sound and complete w.r.t. the configuration semantics.
% This strategy is clearly not satisfactory for some applications such as understanding a domain. It also raises the question of the utility of the resulting AFM compared to the direct use of the configuration matrix.
 The disadvantage is that some constraints in $\Phi$ are likely to be redundant with the attributed feature diagram. 

\begin{equation}\label{eq:phi}
\Phi = \bigvee_{i=1}^{M}{\bigwedge_{j=1}^{N}{(V_j = c_{ij})}} \text{ with N the number of columns and M the number of rows.}
\end{equation}
 
 Ideally, $\Phi$ should express the exact and sufficient set of constraints not expressed in the attributed feature diagram, i.e.,  $\llbracket \Phi \rrbracket =
\llbracket cm \rrbracket \setminus \llbracket FD \rrbracket$.
 Synthesizing a minimal set of constraint may require complex and time-consuming computations. We consider that (1) the development of efficient techniques for simplifying $\Phi$ and (2) the investigation of the usefulness and readability of arbitrary constraints\footnote{We recall that \emph{complex} constraints as defined by the grammar~\ref{fig_rcGrammar} are already part of the synthesis. Arbitrary constraints represent other forms of constraints involving more than two features or attributes -- hence questioning their usefulness or readability by humans.} are both out of the scope of this report.  
 
%  It also raises the question of the utility of the resulting AFM compared to the direct use of the configuration matrix.

% If $\Phi$ is not computed or incomplete, the AFM will probably represent an overapproximation of the configuration matrix.

\myexample{%
The AFM in Figure~\ref{fig:afm-b} exactly represents the configuration matrix of Figure~\ref{fig:afm-a}. To match the configuration semantics of the configuration matrix, the AFM relies on a constraint $\Phi$. This particular constraint cannot be expressed by an attributed feature diagram as defined in Definition~\ref{def:afd}. Therefore, if $\Phi$ is not computed the AFM would represent an over-approximation of the configuration matrix. In particular, the 2 additional configurations of Table~\ref{tab:overapprox} would also be legal for the AFM.
}

This illustrates a trade-off between the expressiveness of the AFD -- especially the expressiveness of $RC$ -- and its over-approximation w.r.t the configuration matrix. 
Quantifying the difference between the two configuration semantics would be an important metric for evaluating the expressiveness of the AFD. We leave it as future work.

\begin{table}
\centering
\small
\begin{tabular}{|c|c|c|c|c|}
\hline 
{\bf LicenseType} & {\bf LicensePrice} & {\bf LanguageSupport} & {\bf Language} & \bf{WYSIWYG} \\ 
\hline 
GPL & 0 & Yes & PHP & Yes \\ 
\hline
GPL & 10 & Yes & PHP & No \\ 
\hline  
\end{tabular} 
\caption{\label{tab:overapprox}An example of the over-approximation of the AFM in Figure~\ref{fig:afm-b} if $\Phi$ is ignored}
\end{table}

\section{Synthesis Algorithm}
\label{sec:approach}
The input of our algorithm is a configuration matrix and some domain knowledge. 
It returns a maximal AFM. The domain knowledge parametrizes the synthesis in order to synthesize a unique AFM.
The synthesis is divided into two parts. First, we synthesize the AFD, then we compute an additional constraint $\Phi$ to ensure the soundess of our algorithm.
In this section, we focus only on the first part. A high-level description of the algorithm for synthesizing an AFD is presented in Algorithm~\ref{alg_synthesis}.
 For the second part, we use Equation~\ref{eq:phi} to compute $\Phi$. 
(As previously stated in Section~\ref{sec:overapprox}, the simplification of $\Phi$ w.r.t. diagram is out of scope of this report.)

\algsetup{
linenodelimiter= 
}
\begin{algorithm}
\caption{\sc \small Attributed Feature Diagram Synthesis}
\begin{algorithmic}[1]
\label{alg_synthesis}
%\small
\REQUIRE A configuration matrix $\mathbf{MTX}$ and domain knowledge $\mathbf{DK}$% about the desired AFM
\ENSURE An attributed feature diagram $\mathbf{AFM}$
\newline
Extract the features, the attributes and their domains\\

\STATE $(F, A, D, \delta) \leftarrow $ extractFeaturesAndAttributes($MTX, DK$) % remove dead features also
%\STATE $(D, \delta) \leftarrow $ extractAttributeDomains($A, MTX, K$)
\newline
\newline
Compute binary implications\\
\STATE $BI \leftarrow $ computeBinaryImplications($MTX$)
\label{ln_bicomp}
\newline
\newline
Define the hierarchy\\
\STATE $(BIG, MTXG) \leftarrow $ computeBIGAndMutexGraph($F, BI$)
\label{ln:big}
\STATE $H \leftarrow $ extractHierarchy($BIG, DK$)
\label{ln:hierarchy}
\STATE $\alpha \leftarrow $ placeAttributes($BI, F, A, DK$)
\label{ln:placeAttributes}
\newline
\newline
Compute the variability information\\
\STATE $E_M \leftarrow $ computeMandatoryFeatures($H, BIG$)
\STATE $FG \leftarrow $ computeFeatureGroups($H, BIG, MTXG, DK$)
%\STATE $MutexGraph \leftarrow $ computeMutexGraph($BI$)
%\STATE $G_{MTX} \leftarrow $ computeMutexGroups($MutexGraph, H$)
%\STATE $G_{OR} \leftarrow $ computeOrGroups($MTX, H$)
%\STATE $G_{XOR} \leftarrow $ computeXorGroups($G_{MTX}, G_{OR}$)%, MTX, H$)
\newline
\newline
Compute cross tree constraints
\STATE $RC \leftarrow $ computeConstraints($BI, DK, H, E_M, FG$)
\label{alg:constraints}
\newline
\newline
Create the attributed feature diagram\\
\RETURN $AFD(F$, $H$, $E_{M}$, $FG$, $A$, $D$, $\delta$, $\alpha$, $RC)$
%\STATE $\Phi  \leftarrow $ computeAdditionalConstraint($MTX, AFD$)
%\RETURN $AFM(AFD,\Phi)$
\end{algorithmic}
\end{algorithm}

%\rbc{
%Our approach for synthesizing feature models has three steps:
%\begin{enumerate}
%\item extract binary cross-attribute implications
%\item construct the hierarchical structure 
%\item link attributes to the features in the hierarchy
%\end{enumerate}
%}
%
%\rb{The first step is done using only the input configuration set. The next two steps take into account the synthesis parameters.} 

\subsection{Extracting Features and Attributes}
The first step of the synthesis algorithm is to extract the features ($F$), the attributes ($A$) and their domains ($D, \delta$).
This step essentially relies on the domain knowledge to decide how each column of the matrix must be represented as a feature or an attribute.
%
%For each feature, its presence or absence in a configuration is detected according to the domain knowledge. 
For each feature, the domain knowledge specifies which values in the corresponding column indicate the presence of the feature, and which values map to the absence of the feature. 
%\myexample{The domain knowledge in Figure~\ref{} defines that the cells in Figure~\ref{fig:afm-a} that contain "Yes" denotes the presence of a feature while the cells that contain "No" denotes the absence of a feature. Another domain knowledge could lead to another interpretation.}
%\gb{ref to domain knowledge}
%\rb{I still don't understand this sentence. How is configuration related to the domain knowledge?
%\gb{I added a simple example. The domain knowledge is required to interpret the values of the cells. Without domain knowledge we cannot be sure that "Yes" means "true"}
%}
%
For each attribute, all the disctinct values of its corresponding column form the first part of its domain ($V_d$). The other parts, the null value $0_d$, and the partial order $<_d$, are computed according to the domain knowledge.
At the end, we discard all the dead features, \ie features that are always absent.

\myexample{Let's consider the variable \feature{LanguageSupport} in the configuration matrix of Figure~\ref{fig:afm-a}. Its domain has only 2 possible values: \textit{Yes} and \textit{No}. According to our knowledge about Wiki engines, these values represents boolean values. Therefore, the variable \feature{LanguageSupport} is identified as a feature. Following the same process, \feature{WYSIWYG} and \feature{LicenseType} are also identified as a feature and the other variables are identified as attributes. 
}

\subsection{Extracting Binary Implications}
An important step of the synthesis is to extract binary implications between features and attributes. To compute all possible implications, we rely on the formal definition of the input configuration matrix (see Definition~\ref{def:matrix}).
We recall that a configuration matrix is a $M \times N$ matrix $\mathbf{C}$.
Each variable $V_j$ of the matrix takes its value from a domain $D_j \in D$ defined by equation~\ref{eq_dom}.

\begin{equation}\label{eq_dom}
D_j = \{c_{i,j} | 1 \leq i \leq M\}
\end{equation}

Intuitively, the domain of variable $V_j$ is the set of all values that appear in the $j$th column of matrix $\mathbf{C}$.

\myexample{
Figure~\ref{fig:afm-a} shows a configuration matrix containing 5 variables and 8 configurations. The first configuration, named \textit{Confluence} is $\{$Commercial, 10, Yes, Java, Yes$\}$. The domain of the variable \feature{LicensePrice} is $\{0,10,20\}$.
}
%\paragraph*{Example}
%Figure~\ref{fig:cm} shows a $3 \times 5$ configuration matrix, containing three different mobile phone configurations from the feature model in Figure~\ref{fig:afm}.
%
%\begin{figure}[H]
% \centering
% 		\includegraphics[width=.49\textwidth]{figs/exampleCM.pdf}
% 	\caption{\label{fig:cm} A sample configuration matrix.}
%\end{figure}

A binary implication $B_{i,j,u,S}$ is a predicate, defined by equation~\ref{eq:binary_implication}, that maps $N$-tuples (e.g., rows of the matrix) to truth values (i.e., true and false). In binary implication $B_{i,j,u,S}$, $i$ and $j$ are integer values in $[1..N]$, $u$ belongs to $D_i$, and $S$ is a subset of $D_j$.
\begin{equation}
\label{eq:binary_implication}
B_{i,j,u,S}(a_1, ..., a_N) = (a_i = u \Rightarrow (a_j \in S))
\end{equation}

We use $BI(\mathbf{C})$ to denote the set of all binary implications that are valid for a configuration matrix $\mathbf{C}$.
\begin{align*}
BI(\mathbf{C}) &= \{B_{i,j,u,S}(a_1, ..., a_N) =  (a_i = u \Rightarrow (a_j \in S))| \\
%& (a_i = u \Rightarrow (a_j \in S \wedge a_j \notin (D_j-S)))| \\
& \qquad 1 \leq i, j \leq N, i \neq j, u \in D_i, S \subseteq D_j, \\
%& \qquad B_{i,j,u,S}: D_1 \times ... \times D_N \mapsto \{0, 1\}, \\
& \qquad \forall 1 \leq k \leq M.~ B_{i,j,u,S}(\mathbf{c}_k) = true
\}
\end{align*}

% ==== Previous implementation ==== %
%Given a configuration matrix $\mathbf{C}$, Algorithm~\ref{alg_constExtract} computes the set of all binary implications. The algorithm iterates over all columns of $\mathbf{C}$, and all distinct values in each column. For each distinct value $u$ of each column $i$, the algorithm computes $\mathbf{C}^{i,u}$, which is a $P \times N $ matrix that contains only those rows of $\mathbf{C}$ for which the $i$th element is equal to $u$:
%
%\begin{equation}
%\mathbf{C}^{i,u} = [\mathbf{c'}_1, ..., \mathbf{c'}_P]^t ~ \text{such that} ~ \forall 1 \leq j \leq P.~ c'_{j,i} = u
%\end{equation}
%
%In line~\ref{ln_dom} of the algorithm, we compute $S$ as the domain of the $j$th column of $\mathbf{C}^{i,u}$. This domain is a nonempty subset of $D_j$. The values $i, j, u$ and $S$ identify a binary implication. In line~\ref{ln_res}, we add the binary implication identified by the quintuple $\langle i, j, u, S \rangle$ to the output set $BI$.
%

Given a configuration matrix $\mathbf{C}$, Algorithm~\ref{alg_constExtract} computes the set of all binary implications. 
First the algorithm initializes the set of binary implications $BI$ with an empty set.
Then, it iterates over all pairs $(i,j)$ of columns and all configurations $c_k$ in $\mathbf{C}$. 
%The objective of the inner loop is to compute the set of values $S(i,j,c_{k,i})$ of column $j$ that are implied when the column $i$ takes the value $c_{k,i}$. 
The objective of the inner loop is to compute the set $S(i,j,c_{k,i})$, which contains the values of those cells in column $j$ for which the corresponding cell in column $i$ is equal to $c_{k,i}$.
%Intuitively, it maps the values of column $i$ to the corresponding values of column $j$.

\algsetup{
linenodelimiter= 
}
\begin{algorithm}
\caption{\sc \small computeBinaryImplications}
\begin{algorithmic}[1]
\label{alg_constExtract}
\REQUIRE A configuration matrix $\mathbf{C}$
\ENSURE A set of binary implications $BI$\\
\STATE $BI \leftarrow \emptyset$

%\FORALL {$1 \leq i \leq N $ and $u \in D_i$}
%\STATE Compute $\mathbf{C}^{i,u}$
%\label{ln_cp}
%\FORALL {$1 \leq j \leq N $ such that $j \neq i$ }
%\STATE $S \leftarrow $ \textsl{domain}(\textsl{col}($\mathbf{C}^{i,u}$, $j$))
%\label{ln_dom}
%\STATE $BI \leftarrow BI \cup \{\langle i, j, u, S \rangle\}$
%\label{ln_res}
%\ENDFOR
%\ENDFOR

\FORALL {$(i,j)$ such that $1 \leq i,j \leq N $ and $i \neq j$}
	\label{ln:outer}
	\FORALL {$c_k$ such that $1 \leq k \leq M $}
		\label{ln:inner}
		\IF{$S(i,j,c_{k,i})$ does not exists}
			\label{ln:test_s}
			\STATE $S(i,j,c_{k,i}) \leftarrow \{c_{k,j}\}$	
			\label{ln:init}
   			\STATE $BI \leftarrow BI \cup \{\langle i, j, u, S(i,j,c_{k,i}) \rangle\}$
			
    	\ELSE
			\STATE $S(i,j,c_{k,i}) \leftarrow S(i,j,c_{k,i}) \cup \{c_{k,j}\}$		    
			\label{ln:add}
	    \ENDIF
	\ENDFOR
\ENDFOR
\RETURN $BI$
\end{algorithmic}
\end{algorithm}

In line~\ref{ln:test_s}, the algorithm tests if $S(i,j,c_{k,i})$ already exists, \ie if $c_{k,i}$ is encountered for the first time. If this is the case, $S(i,j,c_{k,i})$ is initialized with the value of column $j$ for the current configuration $c_k$. Then, a new binary implication is created and added to the set $BI$.
If $S(i,j, c_{k,j})$ already exists, the algorithm simply adds the value of column $j$ for the current configuration $c_k$ to the set $S(i,j,c_{k,i})$.
At the end of the inner loop, $BI$ contains all the binary implications of the pair of columns $(i,j)$.

\subsection{Defining the Hierarchy}
The hierarchy $H$ of an AFD is similar to the hierarchy of a boolean feature model. It is a rooted tree of features such that $\forall (f_1, f_2) \in E, f_1 \Rightarrow f_2$, \ie each feature implies its parent.
As a result, the candidate hierarchies, whose parent-child relationships violate this property can be eliminated upfront.

\myexample{
In Figure~\ref{fig:afm-b}, the feature \feature{GPL} implies its parent feature \feature{LicenseType}.
However, \feature{GPL} could not be a child of \feature{Commercial} as they are mutually exclusive in the configuration matrix of Figure~\ref{fig:afm-a}.
}

To define the hierarchy of the AFD, we rely on the \emph{Binary Implication Graph} (BIG) of a configuration matrix (see Definition~\ref{def:BIG}) to guide the selection of legal hierarchies. The BIG represents every implication between two features of a formula, thus representing every possible parent-child relationships a legal hierarchy can have.
Therefore, we can promote a rooted tree inside the BIG as the hierarchy of the AFD. This step is performed according to the domain knowledge.

To compute the BIG, we simply iterate over the binary implications ($BI$) computed previously.
For each constraint in $BI$ that represents an implication between two features, we add an edge in the BIG.
As we iterate over $BI$, we take the opportunity to compute the mutex graph (see Definition~\ref{def:mutex_graph}) which will be used in the computation of feature groups.

\begin{mydef}[Binary Implication Graph (BIG)]
\label{def:BIG} 
A binary implication graph 
%of a Boolean formula $\phi$ over $\mathcal{F}$ 
of a configuration matrix $\mathbf{C}$
is a directed graph $(V_{BIG}, E_{BIG})$ where $V_{BIG} = F$ and $E_{BIG} = \{(f_i,f_j) \mid `f_i \Rightarrow f_j' \in BI(\mathbf{C})\}$.
\end{mydef}

After choosing the hierarchy of features, we focus on the place of the attributes.
An attribute $a$ can be placed in a feature $f$ if $\neg f \Rightarrow (a = 0_a)$. 
As a result, the candidate features which verify this property are considered as legal positions for the attribute. To compute legal positions, we iterate over the binary implications ($BI$) to check the previous property for each attribute.
Here again, we promote, according to the domain knowledge, one of the legal positions of each attribute.

\myexample{
In Figure~\ref{fig:afm-a}, the attribute \feature{Language} has a domain \textit{d} with "--" as its null value, \ie $0_{d} = \text{"--"}$. This null value restricts the place of the attribute. As defined in Definition~\ref{def:afd}, an attribute can be placed in a feature if the attribute takes its null value when the feature is not selected. This property holds for the attribute \feature{Language} and the feature \feature{LanguageSupport}. However, the configuration \textit{MediaWiki} forbids the attribute to be placed in feature \feature{WYSIWYG}. The value of \feature{Language} is not equal to its null value while \feature{WYSIWYG} is not selected.
}

\subsection{Computing the Variability Information}
As the hierarchy of an AFD is made of features only, attributes do not impact the computation of the variability information (optional/mandatory features and feature groups).
Therefore, we can rely on algorithms that can be applied on boolean FMs.

% mandatory features
% - for each edge in hiearchy, check that the inverse exists in the BIG
First, let's focus on the computation of mandatory features.
In an AFD, mandatory features are implied by their parents.
To check this property we rely on the BIG as it represents every possible implication between two features. 
For each edge $(c, p)$ in the hierarchy, we check that the inverted edge $(p, c)$ exists in the BIG.
If this is the case, we add this edge to $E_M$.

% feature groups
% - we rely on She et al., 2014
% - for mutex groups we take the cliques in the mutex graph
% - for or groups, we use their algorithm 
%	- but it can be time consuming as we will see in evaluation
%   - so it can be deactivated
% - for mutex groups
%	- if or groups are computed, we take groups that are both mutex and or
%	- if not, we check that the parent of a mutex group implies at least one child
For computing feature groups, we reuse algorithms from the synthesis of boolean FMs. In the following, we briefly describe the computation of each type of group. Further details can be found in~\cite{andersen2012,she2014efficient}.

For mutex-groups ($G_{MTX}$), we compute a mutex graph (see Definition~\ref{def:mutex_graph}) that contains an edge whenever two features are mutually exclusive.
As a result, the mutex-groups are the maximum cliques of this mutex graph.
The computation of the mutex graph is performed during the computation of the BIG by iterating over the binary constraints ($BI$).

\begin{mydef}[Mutex Graph]
\label{def:mutex_graph} 
A mutex graph of a configuration matrix $\mathbf{C}$ is an undirected graph $(V_{MTX}, E_{MTX})$ where $V_{MTX} = F$ and $E_{MTX} = \{(f_i,f_j) \mid ` f_i \Rightarrow \neg f_j' \in BI(\mathbf{C})\}$.
\end{mydef}

For or-groups ($G_{OR}$), we translate the input matrix to a binary integer programming problem. Finding all solutions to this problem results in the list of or-groups.

For xor-groups ($G_{XOR}$), we simply list the groups that are both mutex and or-groups.
However, the computation of or-groups might be time consuming as we will see in Sections~\ref{sec:theoretical} and~\ref{sec:practical}. If this computation is disabled, we offer an alternative technique for computing xor-groups. The alternative technique consists in checking, for each mutex-group, that its parent implies the disjunction of the features of the group. 
For that, we iterate over the binary implications ($BI$) until we find that the property is inconsistent.
%\gb{explain that we check that (p and not (c1 or ... or cn)) is inconsistent by iterating over BI}

To ensure the maximality of the resulting AFM, we discard any mutex or or-groups that is also an xor-group.
%Moreover, the remaining features groups may overlap. In that case, we use the domain knowledge to decide which groups we keep in the synthesized AFD.

Finally, the features that are not involved in a mandatory relation or a feature group are considered optional.

\subsection{Computing Cross Tree Constraints}
The final step of the AFD synthesis algorithm is to compute cross tree constraints ($RC$) in order to further restrict the configuration semantics of the AFD.
We generate 3 kinds of constraints: \textit{requires}, \textit{excludes} and \textit{complex} constraints. 

% requires -> edges in BIG that are not represented in the hierarchy + mandatory features
A requires constraint represents an implication between two features. All the implications contained in the BIG (\ie edges) that are not represented in the hierarchy or mandatory features, are promoted as requires constraints.

\myexample{In Figure~\ref{fig:afm-b}, the implication \feature{Commercial} $ \Rightarrow $ \feature{Java} is not represented in the hierarchy nor in mandatory features. As a consequence, this constraint is added to $RC$ and appears below the hierarchy of the FM.}

% excludes -> edges in MTX graph that are not represented in mutex groups
Excludes constraints represent the mutual exclusion of two features. Such constraints are contained in the mutex graph. As the previously computed mutex-groups may not represent all the edges of the mutex graph, we promote the remaining edges of the mutex graph as excludes constraints.

\myexample{In Figure~\ref{fig:afm-b}, the features \feature{NoLimit} and \feature{LanguageSupport} are mutually exclusive but they are not part of a mutex group. Therefore, the excludes constraint \feature{NoLimit} $ \Rightarrow \neg$\feature{LanguageSupport} is added to $RC$ in order to represent this relation.}

% complex ctc -> merge BI according to a value k provided by domain knowledge
Finally, complex constraints are all the constraints following the grammar described in Figure~\ref{fig_rcGrammar} and involving attributes.
First, we transform each constraint refering to one feature and one attribute to a constraint that respects the grammar of $RC$.
Then, we focus on constraints in $BI$ that involve two attributes and we merge them according to the domain knowledge. 
%The form of the constraints in $BI$ is $(a_i = u \Rightarrow (a_j \in S))$. 

The domain knowledge provides the information required for merging binary implications as $(a_i, k)$ pairs, where $a_i$ is an attribute, and $k$ belongs to $D_i$. Using the pair $(a_i, k)$, we partition the set of all binary implications with $a_i =u$ on the left hand side of the implication into three categories:  those with $u < k$, those with $u = k$, and those with $u >k$. Let $b_{j,1}, b_{j,2}, ..., b_{j,p}$ be all such binary implications, belonging to the same category, and involving $a_j$ (i.e., each $b_{j,r}$ is of the form $(a_i = u_r \Rightarrow a_j \in S_r)$). We merge these binary implications into a single one: $(a_i \in \{u_1, u_2, ..., u_p\} \Rightarrow a_j \in S_1 \cup S_2 \cup ... \cup S_p)$.
%
%First, we separate constraints such that $a_i=k$, $a_i<k$ and $a_i>k$, with $k$ a bound provided by the domain knowledge.
%Then, we merge the constraints in each category.
%Two constraints $(a_{i1} = u_1 \Rightarrow (a_{j1} \in S_1))$ and $(a_{i2} = u_2 \Rightarrow (a_{j2} \in S_2))$ can be merged if they refer to the same variables, \ie $a_{i1} = a_{i2} = a_i$ and $a_{j1} = a_{j2} = a_j$. It results in a new constraint $(a_i = (u_1 \cup u_2) \Rightarrow (a_j \in S_1 \cup S_2))$ with $u_1$ and $S_1$, the values of the first constraint (resp. $u_2$ and $S_2$ for the second constraint).
%
Finally, if the merged constraints can be expressed with the grammar of $RC$, we add them to $RC$. Otherwise, we discard them.

\myexample{%
From the configuration matrix of Figure~\ref{fig:afm-a}, we can extract the following binary implication: $\feature{GPL} \Rightarrow \feature{LicensePrice} \in \{0, 10\}$. 
We also note that the domain of \feature{LicensePrice} is $\{0, 10 ,20\}$. Therefore, the right side of the binary implication can be rewriten as \feature{LicensePrice} $ <= 10$. As this constraint can be expressed by the grammar of $RC$, we add \feature{GPL} $ \Rightarrow $ \feature{LicensePrice} $ <= 10$ to $RC$ as shown in Figure~\ref{fig:afm-b}.
}

%\section{Evaluation} 
%\label{sec:evaluation}
%\input{evaluation.tex}

\section{Theoretical Evaluation}
\label{sec:theoretical}
Our algorithm presented in Section~\ref{sec:approach} addresses the synthesis problem defined in Definition~\ref{def:problem}. This definition states that the synthesized AFM must be maximal and have the same configuration semantics as the input configuration matrix. In the following, we check these two properties. We also evaluate the scalability of our algorithm by analyzing its theoretical time complexity.

\subsection{Soundness and Completeness}

Synthesizing an AFM that represents the exact same set of configurations (\ie configuration semantics) as the input configuration matrix is primordial. In order to ensure this property, the synthesis algorithm must be sound (see Definition~\ref{def:soundness}) and complete (see Definition~\ref{def:completeness}).

%Synthesizing an AFM that represents the exact same set of configurations (\ie configuration semantics) as the input configuration matrix is primordial. 
%If the AFM is too permissive, it would expose the user to illegal configurations.
%To prevent this situation, the algorithm must synthesize a sound AFM(see Definition~\ref{def:soundness}).
%Conversely, If the AFM is too constrained, it would prevent the user from selecting available configurations, resulting in unused variability.
%Therefore, the algorithm must also synthesize a complete AFM (see Definition~\ref{def:completeness}).
%Figure~\ref{fig:soundness_completeness} illustrates how these two properties are related to the configuration semantics of the input configuration matrix.

\subsubsection{Valid and Comprehensive Computation of Binary Implications}
\label{sec:bi}

An AFM can be seen as the representation of a set of configurations (see Defintion~\ref{def:configuration_semantics}) but also as a set of constraints over the features and attributes.
These two views are equivalent but present two opposite reasonings. The former view focuses on adding configurations while the latter focuses on removing configurations. In the following, we will switch from one view to the other in order to prove the soundness and completeness of our algorithm.

A central operation in our algorithm is the computation of all the binary implications that are present in the input configuration matrix. The soundness and completeness of our algorithm are tied to a valid and comprehensive computation of these constraints. Definition~\ref{def:soundness_bi} and Definition~\ref{def:completeness_bi} define these two properties for the computation of binary implications, respectively. 
%As explained previously, the view of an AFM as a set of constraints leads to an opposite reasoning compared to the view as a set of configurations. As a consequence, the soundness of the binary implications is related to the completeness of our algorithm and vice versa.

\begin{mydef}[Validity of binary implications]
\label{def:soundness_bi}
Let $\mathbf{C}$ be an $M \times N$ configuration matrix.  A set of binary implications $BI$ is valid w.r.t. $\mathbf{C}$, iff for each $\langle i, j, u, S \rangle \in BI$ the following condition holds:
\begin{equation}
\label{eq:soundness_bi}
\forall k \in  [1..M]. ~ (c_{k, i} = u) \Rightarrow c_{k, j} \in S % \wedge (c_{k, j} \notin (D_j - S))
\end{equation}
\end{mydef}

%A set of binary implications is sound w.r.t. $\mathbf{C}$, iff all binary implications in it are sound with respect to $\mathbf{C}$. 

\begin{mydef}[Comprehensiveness of binary implications]
\label{def:completeness_bi}
Let $\mathbf{C}$ be an $M \times N$ configuration matrix.  A set of binary implications $BI$ is comprehensive w.r.t. $\mathbf{C}$, iff:
\begin{equation}
\label{eq:completeness_bi}
\forall i, j \in  [1..N], ~ \forall u \in D_i. ~ \exists ~ \langle i, j, u, S \rangle \in BI ~\textnormal{such that}~ S \subseteq D_j: 
\end{equation}
Intuitively, for each possible combination of $i, j, $ and $u$, at least one binary implication exists in $BI$. 
\end{mydef}

%\begin{theorem}\label{th_alg_prop}
Let $\mathbf{C}$ be an $M \times N$ configuration matrix. The set of binary implications $BI$ computed using Algorithm~\ref{alg_constExtract} is valid and comprehensive with respect to $\mathbf{C}$. 
%\end{theorem}
The proof of these two properties is as follows:

%\begin{proof}
{\bf Validity:} Let $\langle i, j, u, S \rangle$ be an arbitrary binary implication in $BI$. 
In line~\ref{ln:init} and~\ref{ln:add} of the algorithm, we add the value of column $j$ for configuration $c_k$ (noted as $c_{k,j}$) to $S(i,j,c_{k,i})$. This particular $S$ is bound to $i$, $j$ and the value of the column $i$ for the configuration $c_k$ (noted as $c_{k,i}$).
%In line 3 of the algorithm, we compute $\mathbf{C}^{i,u}$ by pruning all rows of $\mathbf{C}$ for which $c_{k, i} \neq u$. Then, in line 5, for each $j$, $S$ is computed as the set of all distinct values in column $j$ of $\mathbf{C}^{i,u}$. 
As a result, for each $k \in [1..M]$ such that $c_{k, i} = u$, $c_{k, j}$ is in $S$, and Equation~\ref{eq:soundness_bi} holds for $\langle i, j, u, S \rangle$.

{\bf Comprehensiveness:} The for statements in lines~\ref{ln:outer} and~\ref{ln:inner} iterate over all values of $i$, $j$ and $u$, and generate all distinct combinations $(i, j, u)$. For each of these combinations, line~\ref{ln:test_s} ensures that one binary implication is added to $BI$. Let $\langle i, j, u, S \rangle$ be a binary implication in $BI$. In addition, line~\ref{ln:init} and~\ref{ln:add} guarantee that the value added to $S$ is in $D_j$ and thus $S \subseteq D_j$. Therefore, Equation~\ref{eq:completeness_bi} holds for $BI$.
%\end{proof}

\subsubsection{Proof of Soundness of the Synthesis Algorithm}
According to Definition~\ref{def:soundness}, the algorithm detailed in Section~\ref{sec:approach} is sound if all the configurations represented by the AFM exist also in the configuration matrix.
To prove this property, it is equivalent to show that the constraints represented by the configuration matrix are included in the constraints represented by the AFM.

In general, the attributed feature diagram (see Definition~\ref{def:afd}) is not expressive enough to represent the constraints of the configuration matrix (see Section~\ref{sec:overapprox}). To keep the algorithm sound, we rely on the additional constraint $\Phi$ of the AFM. A simple strategy is to use Equation~\ref{eq:phi} to specify the whole configuration matrix as a constraint, and include it in the AFM as $\Phi$. Such a constraint would, by construction, restrict the configuration semantics of the AFM to a subset of the configuration semantics of the matrix. Therefore, our algorithm is sound if $\Phi$ is carefully defined. Otherwise, it may represent an overapproximation of the input configuration matrix (see Section~\ref{sec:overapprox}).

\subsubsection{Proof of Completeness of the Synthesis Algorithm}
To prove the completeness of our algorithm, we have to show that all the configurations represented by the configuration matrix exist also in the synthesized AFM (see Definition~\ref{def:completeness}). 
It is equivalent to show that all the constraints represented by the AFM are included in the constraints represented by the configuration matrix.
A first observation is that, apart from or-groups, xor-groups and $\Phi$, the AFM is constructed from a set of binary implications computed by Algorithm~\ref{alg_constExtract} (step~\ref{ln_bicomp} of Algorithm~\ref{alg_synthesis}). %Only the computation of or-groups, xor-groups and the final constraint $\Phi$ do not rely on binary implications.

% BI sound => AFM does not have any restrictions in addition to the ones in the matrix
As proved in Section~\ref{sec:bi}, the computed binary implications are valid. Therefore, all the constraints presented by these implications are included in the configuration matrix.
All the subsequent computations relying on the binary implications do not introduce new constraints. They only reuse and transform these binary implications into feature modeling concepts.

For or-groups and xor-groups, we reuse existing techniques that rely on the input configuration matrix. These techniques are designed for a similar context of FM synthesis and ensure the completeness of the resulting FM~\cite{andersen2012,she2014efficient}. 
In our context of AFMs, the features and the attributes are separated in two distinct sets, repesctively $F$ and $A$. Moreover, the attributes do not enter in the definition of the hierarchy $H$ or the feature groups. As a consequence, the adaptation of techniques for computing or-groups and xor-groups to the context of AFMs is straightforward and guarantees the completeness.
%\rb{The last sentence, in my opinion, needs a proof, or at least some form of justification. In particular, in the case of AFM, we are splitting the variables in the matrix into two sets: $F$ and $A$. We should also argue that this partitioning and its implications (attributes don't get nodes in the tree, but may be included in some binary implications, etc) do not affect the properties of the algorithms in~\cite{andersen2012,she2014efficient}.}

%Finally, the computation of the additional constraint $\Phi$ of an AFM is not addressed in this paper but a simple strategy is to translate the configuration matrix as a constraint. By definition, this strategy ensures that the constraints represented by $\Phi$ are included in the constraints represented by the configuration matrix.

Finally, $\Phi$ as computed by Equation~\ref{eq:phi} is equivalent to the configuration matrix. Therefore, adding such a $\Phi$ to AFM does not exclude from it any of the configurations represented by the configuration matrix.

Overall, each step of the algorithm of Section~\ref{sec:approach} guarantees the property of completeness. Therefore, the whole synthesis process ensures completeness.

\subsection{Maximality}
Another challenge is to ensure the maximality of the synthesized AFM in order to avoid trivial answers to the AFM synthesis problem (see Section~\ref{sec:approach}). In this section, we prove that Algorithm~\ref{alg_synthesis} produces a maximal attributed feature model. This is stated in Theorem~\ref{theorem_max}.

\begin{theorem}[Maximality of the configuration semantics]
\label{theorem_max}
Let $M$ be a configuration matrix, and $AFM$ be a sound and complete synthesis of $M$, generated using Algorithm 1. Then, $AFM$ is also maximal.
\end{theorem}

\begin{proof} %
To prove the maximality of $AFM = \langle F, H, E_M, G_{MTX}, G_{XOR}, G_{OR}, A, D, \delta\alpha, RC \rangle$, we show that all conditions in Definition~\ref{def:maximal_afm} hold for $AFM$:

\textbf{$H$ connects every feature in $F$.}
The algorithm ensures that the synthesized hierarchy contains all the features of the configuration matrix. 
The hierarchy is a rooted tree which is extracted from the binary implication graph (BIG) in step~\ref{ln:hierarchy} of Algorithm~\ref{alg_synthesis}. 
By construction (see Definition~\ref{def:BIG}), the BIG contains all the features.
Moreover, the comprehensive computation of the binary implications (see Section~\ref{sec:bi}) ensures that the BIG represents all possible implications of the configuration matrix. 
We can define a hierarchy of an AFM as a spanning tree of the BIG. 
Therefore, every spanning tree of the BIG is a possible hierarchy for the AFM and every possible hierarchy is a spanning tree of the BIG~\cite{becanESE2015}.
%The extraction of a hierarchy consists in selecting a rooted tree that contains all the nodes of the graph, \ie all the features.
The existence of a single spanning tree is not generally ensured for a binary implication graph computed from any arbitrary configuration matrix. Specifically, the BIG may contain more than a single connected component. 
In such cases, we simply create a new root feature $r$ in the BIG, and connect every feature in $F$ to $r$. As a consequence, line 3 of the algorithm always generates a connected graph. 
The \textsf{extractHierarchy} routine in line 4 of Algorithm 1, uses the domain knowledge to chose one spanning tree of the BIG as $H$.
Therefore, $H$ is guaranteed to be a tree containing all features in $F$.
%Therefore, it is always possible to extract a hierarchy that connects all the features in $F$, from the BIG in line 4 of the algorithm.

\textbf{Adding an edge to $E_m$ changes the configuration semantics of $AFM$ ($\llbracket AFM \rrbracket$).}
The algorithm ensures that no edge can be added to the set of mandatory edges, $E_m$, without changing the configuration semantics of the AFM (i.e., adding or removing items from $\llbracket AFM \rrbracket$).
$E_m$ represents every edge $(f_1,f_2)$ of the hierarchy $H$ such that $f_1 \Leftrightarrow f_2$.
In line 6 of the algorithm, for each edge $(f_1,f_2)$ in the hierarchy $H$, we add an edge $(f_2,f_1)$ to $E_m$, iff $(f_2,f_1)$ exists in the BIG.
As mentioned previously, the BIG exactly represents all possible implications of the configuration matrix.
Adding any other edge to $E_m$ implies that the added edge does not exist in BIG. 
Therefore, no edge can be further added to $E_m$ without changing the configuration matrix, which corresponds to the configuration semantics of the AFM.

%As proved in Section~\ref{sec:bi}, the binary implications are complete w.r.t the configuration matrix. Therefore, the binary implication graph is also complete. As the algorithm promotes all bi-implications in the graph as mandatory edges, the maximality is preserved.

\textbf{Adding an item to any of $G_{MTX}$, $G_{OR}$, or $G_{XOR}$ changes $\llbracket AFM \rrbracket$.}
The computation of feature groups is entirely performed with existing techniques that are designed for the synthesis of boolean FMs. 
Our definition of a maximal AFM (see Definition~\ref{def:afd} and~\ref{def:maximal_afm}) preserves the definition of feature groups and the notion of maximality used in boolean FMs.
Moreover, the attributes are clearly separated from features and are not part of the feature groups.  Therefore, the techniques designed for boolean FMs are not impacted by the attributes and keep their properties (soundness, completeness and maximality).
As we are using techniques that are designed for synthesizing sound, complete and maximal boolean FMs~\cite{andersen2012,she2014efficient}, no group can be added to $G_{MTX}$, $G_{OR}$ $G_{XOR}$ without changing the configuration semantics of the $AFM$.

\textbf{Moving any item from $G_{MTX}$ or $G_{OR}$ to $G_{XOR}$ changes $\llbracket AFM \rrbracket$.}
An xor-group is a feature-group that is both a mutex-group and an or-group. 
In Algorithm~\ref{alg_synthesis}, we promote every mutex-group that is also an or-group as an xor-group. 
Moreover, the previous property of maximality ensures that all the possible feature groups are computed.
Therefore, moving another group from $G_{MTX}$ or $G_{OR}$ to $G_{XOR}$ is impossible without changing the configuration semantics of the AFM.

\textbf{Adding a non-redundant constraint to $RC$ changes $\llbracket AFM \rrbracket$.}
As shown in the grammar of $RC$ in Figure~\ref{fig_rcGrammar}, all readable constraints are implications between values of features or attributes. 
We compute the set of all such binary implications in line 2 of the algorithm. As proven in Section~\ref{sec:bi}, Algorithm~\ref{alg_constExtract} perform a comprehensive computation of binary implications in BI. Any additional constraint is therefore, either redundant, forbidden by the domain knowledge (\ie its numerical literal is not part of the interesting values) or inconsistent with configuration matrix $M$. The latter implies that adding a non-redundant constraint to $RC$ would change the configuration semantics of $AFM$.

%The constraints defined as $RC$ in our formalism can be expressed as binary implications.
%The computation of $RC$ consists in rewriting all the binary implications computed by Algorithm~\ref{alg_constExtract} that match the grammar of $RC$ (see Figure~\ref{fig_rcGrammar}). 
%As proven in Section~\ref{sec:bi}, Algorithm~\ref{alg_constExtract} is complete. By construction, the computation of $RC$ is also complete.
%Therefore, a new constraint for $RC$ would be either redundant or change the configuration semantics of the AFM.

Overall, every property of the maximality is verified in the resulting AFM of our algorithm.

\end{proof}

\subsection{Complexity analysis}
\label{sec:complexity_analysis}
The manual elaboration of an AFM is error-prone and time-consuming. As we showed previously, the soundness and completeness of our algorithm address the first problem. For the second problem, we have to evaluate the time complexity of our algorithm. 
It depends on 3 characteristics of the input configuration matrix:
\begin{itemize}
\item number of variables (features + attributes): $v = f + a$
\item number of configurations: $c$
\item maximum domain size (\ie maximum number of distinct values for a variable): $d$
\end{itemize}

In the following, we analyze the complexity for each step of Algorithm~\ref{alg_synthesis} and~\ref{alg_constExtract}.

\subsubsection{Extracting Features and Attributes}
%extract domain values: iterate over each cell of the matrix - $O(v.c)$\\
%extract features and attributes: \\
%- call to knowledge: out of scope (but $O(v.d)$ for the evaluation)\\
%- dead feature: check domain $O(v.d)$. As features have only 2 values in $d$, it is $O(v)$\\
%total: $O(v.c)$

The first step of Algorithm~\ref{alg_synthesis} is the extraction of features, attributes and their domains. This extraction simply consists in running through the configuration matrix to gather the domains and identify features and attributes. The size of the matrix is characterized by the number of variables (columns) and the number of configurations (rows). Therefore, the complexity is $O(v.c)$.

\subsubsection{Extracting Binary Implications}
\label{sec_CompBI}
The extraction of binary implications is computed by Algorithm~\ref{alg_constExtract} whose complexity is as follows.
The outer for loop of lines~\ref{ln:outer}-\ref{ln:add} iterates over all pairs of variables in the configuration matrix ($O(v^2)$).
The inner for loop of lines~\ref{ln:inner}-\ref{ln:add} iterates over all configurations ($O(c)$).
In this loop, we perform 2 types of operations: checking the existence of the set $S(i,j,c_{k,i})$ in line \ref{ln:test_s} and adding one element to a set.
Checking the existence of $S(i,j,c_{k,i})$, for a given $i$ and $j$, involves searching for an element in a map that has potentially $d$ elements. Therefore, its complexity is $O(d)$.
The latter operation is performed on a set that accesses its elements by a hash function. Adding an element to such a set has a complexity of $O(1)$.
%For both operations, we use data structures that access their elements by a hash function. This offers a linear time complexity ($O(n)$).
Overall, the complexity of the computation of binary implications is $O(v^2.c.d)$.

%Let $\mathbf{C}$ be an $M \times N$ configuration matrix ($M$ is the number of configurations, and $N$ is the number of configurable parameters), and $D_1, ..., D_N$ be the domains defined as in equation~\ref{eq_dom}. 
%The for-loop in lines 2-6 iterates over all the values in domains $D_1, ..., D_N$. In each iteration, first $\mathbf{C}^{i,u}$ is computed. Then the domains of all of its columns are calculated. 

%Let $K$ be the size of the largest domain $D_i$, then the for-loop in line 2 repeats $O(N.K)$ times. Computing $\mathbf{C}^{i,u}$, in line 3, is $O(M)$. The for-loop in line 4 repeats $O(N)$ times, and the cost of computing $S$ (line 5), and updating $BI$ (line 6), are $O(M.K)$ and $O(1)$, respectively. The overall worst case time-complexity of Algorithm~\ref{alg_constExtract} is, therefore, $O(N.K.(M+ N.(M.K+1)))$, or $O(N^2.K^2.M)$.

%\rbc{We have implemented Algorithm~\ref{alg_constExtract} using the \emph{clpfd} library of SICStus Prolog~\cite{SICStusPage}. In particular, we use the table constraint, which in SICStus, is implemented using an MDDc propagator~\cite{ChengY10} of time-complexity $O(N^2.K.M)$ \textbf{(Or $O(N^2.K^2.M)$?)}. }

\subsubsection{Defining the Hierarchy}
%The number of binary implications is maximum $v^2.d$\\
%BIG + mutex graph: $O(v^2.d)$\\
%extract hierarchy: out of scope (but $O(f^2log(f))$ with tarjan algorithm for the evaluation)\\
%place attributes: iterate over constraints and check that $\neg f \Rightarrow O_d$, which results in $O(v^2.d^2)$\\
%total: $O(v^2.d^2)$

In line~\ref{ln:big}, we iterate over all the binary implications to compute the binary implication graph and the mutex graph. For each binary implication, if it represents an implication between two features, we add an edge to the binary implication graph. If the binary implication represents a mutual exclusion between two features, we add an edge to the mutex graph. Both checks are done in a constant time. Therefore, the complexity of this step depends on the number of binary implications. As explained in Algorithm~\ref{alg_constExtract}, we create a binary implication for each pair of variables $(i,j)$ and for each value $c_{k,i}$ in the domain of the $i$. It results in a maximum of $v^2.d$ binary implications. As a consequence, the complexity of line~\ref{ln:big} is $O(v^2.d)$.
 
In line~\ref{ln:hierarchy}, we compute the hierarchy of the AFM by selecting a tree in the binary implication graph. This step is performed by the domain knowledge (\ie by a user or an external algorithm). In~\cite{becanESE2015}, we propose an algorithm for this task which has a time complexity of $O(f.log(f^2))$.

In line~\ref{ln:placeAttributes}, we compute all the possible places for each attribute. We recall that an attribute $a$ can be placed in the feature $f$ if $\neg f \Rightarrow (a = 0_a)$, with $0_a$ being the null value of the domain of $a$.
We first initialize the possible places for all attributes as an empty set.
Then, we iterate over each binary implication $\{\langle i, j, u, S\rangle\}$ in order to add the valid places. The loop has a complexity of $O(v^2.d)$ as it represents the maximum number of binary implications.
A valid place is identified as follows.
First, we check that $i$ is a feature, $j$ is an attribute and that $u$ represents the absence of the feature $i$.
Then, we verify that $S$ is equal to the null value of the domain of $j$.
If all these properties are respected, the feature $i$ is thus a valid place for the attribute $j$ and can be added to its possible places.
These verifications have a complexity of $O(d$).
Overall, the complexity of this step is $O(v^2.d^2)$.
% === Old version of the computation of the attributes' possible places === %
%We first initialize the possible places for all attributes as the set of features $F$.
%Then, we iterate over each binary implication $\{\langle i, j, u, S\rangle\}$ in order to remove invalid places. The loop has a complexity of $O(v^2.d)$ as it represents the maximum number of binary implications.
%An invalid place is identified as follows.
%First, we check that $i$ is a feature, $j$ is an attribute and that $u$ represents the absence of the feature $i$.
%Then, we verify that $S$ is not equal to the null value of the domain of $j$.
%If all these properties are respected, the feature $i$ is thus an invalid place for the attribute $j$ and can be removed from its possible places.
%These verifications have a complexity of $O(d$).
%Overall, the complexity of this step is $O(v^2.d^2)$.

\subsubsection{Computing the Variability Information}
%\gb{draft}
%mandatory features: iterate over the edges of the hierarchy ($f-1$ edges as it is a tree) and check the presence of inverted edges in the BIG: The complexity depends on the implementation of the graph library but $O(f^3)$ in really worst case.\\
%mutex groups: find maximal cliques : complexity of Born Kerbosch algorithm $O(3^{f/3})$\\
%or groups: NP-complete\\
%xor groups: there are at most $f$ mutex groups and $f$ or groups: $O(f^2)$\\
%alternative xor groups: there are at most $f$ mutex groups and we check a property for each group by iterating over the configuration: $O(f.c)$\\
%process overlapping groups: \\
%- iterate over the different groups to detect overlapping groups:\\
%- ask domain knowledge to choose: out of scope (but $O(1)$ in the evaluation)
%total: $O()$ with or groups and $O()$ without
The computation of variability information is done in two steps: mandatory features and feature groups.
For detecting mandatory features, we iterate over every edge of the hierarchy. As the hierarchy is a tree containing all the features as nodes, the number of edges is exactly equal to $f-1$. For each edge, we check that the inverted edge exists in the binary implication graph. This check is performed with a complexity of $O(f)$ in our implementation. 
Overall, the complexity is $O(f^2)$.

For feature groups, we rely on existing techniques that are developed in~\cite{andersen2012,she2014efficient}. We summarize the results of the complexity analysis of these techniques.

The computation of mutex groups consists in finding maximal cliques in the mutex graph. As the mutex graph contains all the features as nodes, the complexity of this operation is $O(3^{f/3})$~\cite{Tomita200628}.

For or-groups, the algorithm relies on a binary integer programming which is an NP-complete problem.

For the computation of xor-groups, there are two alternatives. The first one assumes the computation of or-groups. For each or-group, it iterates through the mutex-groups to check if there exists an equivalent group (\ie same parent and same children).
As the number of groups is bounded by the number of features, the two iterations have a complexity of $O(f^2)$. Checking if two groups are equivalent consists in checking the equality of two sets. 
The maximum size for a group is also bounded by the number of features. The check is performed with a complexity of $O(f^2)$.
%In addition, our implementation uses sets that access their elements by a hash function, which guarantees a complexity of $O(f)$. 
Overall, this first technique for computing xor-groups is $O(f^4)$.

The second technique do not assume the computation of or-groups. It iterates over every mutex-group and checks that, for each configuration, the parent of the group implies the disjunction of the features of the group. Checking this property depends on the size of the mutex-group, which is bounded by the number of features $f$.
Overall, the complexity of this second technique for computing xor-groups is $O(f^2.c)$

\subsubsection{Computing Cross Tree Constraints}
\label{sec_CompRC}
The computation of cross tree constraints ($RC$) is performed in three steps.

First, the \textit{requires} constraints are extracted from the binary implication graph. The algorithm iterates over every edge of the graph and checks if it exists either in the mandatory features or in the hierarchy. The binary implication graph is composed of the features as nodes and thus contains at most $f^2$ edges. In addition, the number of edges represented by the mandatory features and the hierarchy is bounded by the number of features $f$. Therefore, the complexity of this step is $O(f^3)$.

Then, the \textit{excludes} constraints are extracted from the mutex graph by looking for all the edges that are in the mutex graph but not in any mutex-group. The algorithm iterates over every edge of the mutex graph and checks that it does not exist in the mutex-groups. Like the binary implication graph, the mutex graph contains at most $f^2$ edges. Moreover, the mutex-groups represent at most $f$ edges (size of the hierarchy). Therefore, checking if an edge exists in the mutex-groups has a complexity of $O(f)$. Overall, the complexity of this step is $O(f^3)$.

Finally, the complex constraints are extracted from the binary implications computed in line~\ref{ln_bicomp} of Algorithm~\ref{alg_synthesis}. 
This extraction is done in two steps. 
% The first step is the main part of the complex constraint computation
In the first step, the algorithm iterates over every binary implication, which results in a complexity of $O(v^2.d)$.
Then, binary implications that refer to a particular pair of attributes (e.g., $(a_i, a_j)$) are divided into three categories based on the bound specified on the first attribute, by the domain knowledge. The items in each category are then merged. 
To speed up the detection of implications that need to be merged, we store the implications in a map indexed by the pairs of attributes. Therefore, the algorithm is working on three maps that have at most $v^2$ elements. Thus, the complexity of an operation on the map is $O(v^2)$.
Merging two binary constraints $(a_i, u_1, a_j, S_1)$ and $(a_i, u_2, a_j, S_2)$ consists in computing the union of $S_1$ and $S_2$. These sets contain values of the attribute $a_j$ and have a maximum size of $d$. Therefore, the complexity for merging two binary constraints is $O(d)$ and the total complexity of this first step is $O(v^2.d(v^2 + d)) = O(v^4.d + v^2.d^2)$. 

% The second step is not that important but still present in the implementation
In the second step of the extraction of constraints, we iterate over the three maps previously computed. Each map has potentially $v^2$ elements. The complexity of the loop is thus $O(v^2)$.
In these maps, the value of each element ($a_i, a_j$) is composed of the result of the previous merge operations, which is a set $S$ of at most $d$ elements.
We check that $S$ can be represented as ``$a_j$ \textsf{OP} $k$'', with \textsf{OP} a comparison operator from the grammar of $RC$ and $k$ a value of the domain of $a_j$, specifying the bound on it.
If this representation is possible, we add it as a complex constraint to $RC$. In our implementation this check has a complexity of $d^2$. % The check mainly consists in sorting the domain of a_j
Therefore, the complexity of the second step is $O(v^2.d^2)$.

Overall, the complexity of the computation of complex constraints is $O(v^4.d + v^2.d^2)$.

%The algorithm iterates over all $(a_i,k)$ pairs in the domain knowledge, where $a_i$ is an attribute, and $k$ is an interesting bound for it. This has a worst-case complexity of $O(v)$, as $v$ provides an upper bound for the number of attributes.  
%Then, binary implications that refer to the same variables are merged according to the domain knowledge\footnote{In this step, we use the domain knowledge to avoid en exhaustive search of complex constraints. As a consequence, it significantly reduces the complexity of the algorithm for computing complex constraints.}. 
%Then, binary implications that refer to a particular attribute/variable (e.g., $a_i$) are divided into three categories based on the value of $k$, and the ones in each category are merged. 
%To speed up the detection of implications that need to be merged, we store the implications in a map that accesses its elements by a hash function.
%The number of binary implications that have different variables are at most $v^2$. The complexity of a lookup in the map is therefore $O(v^2)$.
%The number of binary implications that involve a particular attribute/variable is at most $v.d$. The complexity of a lookup in the map is therefore $O(v.d)$.
%Merging two binary implications consists in merging two sets of values, each containing at most $d$ values.
%In our implementation, it is done with a complexity of $O(d)$ as we use a set with direct access to its elements.
%Overall, the complexity of this step is $O(v(v.d+d)) = O(v^2.d)$.

\subsubsection{Overall Complexity}
The time complexity analysis of Algorithm~\ref{alg_synthesis} shows that, apart from the computation of mutex-groups and or-groups, the synthesis of an attributed feature diagram has a polynomial time complexity. In particular, its complexity is $O(v^4.d + v^2.d^2 + v^2.c.d)$.
The complexity of the computation of mutex-groups and or-groups are exponential and NP-complete, respectively. They represent the hardest parts of Algorithm~\ref{alg_synthesis} from a theoretical point of view.
To obtain an AFM, we have to take into account the computation of $\Phi$. With our algorithm in Equation~\ref{eq:phi}, we simply iterate over each value of the configuration matrix, which results in a complexity of $O(v.c)$.

\section{Empirical Evaluation}
\label{sec:practical}
%\subsection{Research Questions}
%For this practical evaluation, we evaluate the AFD synthesis algorithm described in Algorithm~\ref{alg_synthesis} over 3 research questions:
%\begin{description}
%\item[RQ1] \textbf{Scalability}: How does the synthesis algorithm scale w.r.t. numbers of variables (features and attributes), numbers of configurations, and maximum domain size?
%
%\item[RQ2] \textbf{Overapproximation}: How does the synthesis algorithm retrieve the original configuration set? To what extent synthesized constraints are an overapproximation of the original configurations? We provided theoretical evidence that we compute an overapproximation, but practically, what is the impact and importance of the phenomena? 
%
%\item[RQ3] \textbf{Parameterization}: Can the synthesis algorithm retrieve the original attributed feature model? Numerous equivalent attributed feature models can characterize the same configuration set. We expect our algorithm can be parametrized to obtain an adequate attributed feature model. 
%\end{description}

For this practical evaluation, we evaluate the time complexity of the AFD synthesis algorithm described in Algorithm~\ref{alg_synthesis}.
In the following experiments, the algorithm takes random matrices as input. We evaluate its scalability over 3 caracteristics of the input: number of features and attributes, number of configurations and the maximum number of distinct values for an attribute.

\subsection{Experimental Settings}
\label{subsec:settings}
For this evaluation, we use randomly generated matrices as input.
%Inputs: 
%\begin{itemize}
%\item randomized matrix
%\item attributed feature models 
%\item product comparison matrices
%\end{itemize}
%
% Random matrix generator
Our random matrix generator has 3 parameters :
\begin{itemize}
\item number of variables (features and attributes)
\item number of configurations
\item maximum domain size (\ie maximum number of distinct values in a column)
\end{itemize}

The type of each variable (feature or attribute) is randomly selected according to a uniform distribution. 
An important property is that our generator does not ensure that the number of configurations and the maximum domain size are reached at the end of the generation. Any duplicated configuration or missing value of a domain is not corrected. Therefore, the parameters entered for our experiments may not reflect the real properties of the generated matrices. To avoid any misinterpretation or bias, we present the concrete numbers in the following results.

Moreover, to reduce fluctuations caused by the random generator, we perform the experiments at least 100 times for each triplet of parameters.
In order to get the results in a reasonable time, we used a cluster of computers. Each node in the cluster is composed of two Intel Xeon X5570 at 2.93 Ghz with 24GB of RAM.
% paradent cluster : 2 * Intel Xeon L5420 2.5 GHz / 6MB with 32GB of RAM
% parapide cluster : 2 * Intel Xeon X5570 2.93 GHz / 6MB with 24GB of RAM

\subsection{Expected practical complexity}
\label{sec:expected_complexity}
The analysis of the theoretical time complexity of Algorithm~\ref{alg_synthesis} shows two hard parts (see Section~\ref{sec:complexity_analysis}).
The first hard part is the computation of mutex-groups, which has an exponential worst-case complexity.
In~\cite{she2014efficient}, She \etal show that the algorithm for computing mutex-groups takes a few seconds for FMs with up to 290 features and can scale to FMs with 5000 features. 
We also note that our dataset produces mutex graphs with 6 edges in average and 93.8\% of the mutex graphs contain absolutely no edges. In such cases, computing maximal cliques (\ie computing mutex-groups) is trivial. Therefore, we expect a non significant impact of the computation of mutex-groups on the practical complexity.

The second hard part of the algorithm is the computation of or-groups, which is NP-complete.
Unlike mutex groups, She \etal show that the computation of or-groups hardly scales~\cite{she2014efficient}. 
In Section~\ref{sec:experiment_or}, we confirm this observation by evaluating the scalability of the computation of or-groups w.r.t the number of variables.
Based on this observation, we deactivated the computation of or-groups in the rest of the experiments.

For the rest of the algorithm, the theoretical complexity is polynomial: $O(v^4.d + v^2.d^2 + v^2.c.d)$.
We note that some operations in the algorithm (\eg the computation of complex constraints) rely on the use of data structures based on hash functions. An operation on such data structure has a linear worst case complexity. However, the probability of reaching the worst case complexity is very low in practice. 
Moreover, the size of the mutex graph also impacts the practical complexity of the computation of excludes constraints which has a theoretical complexity of $O(f^3)$.
The effective size of other data structures used during the algorithm may further reduce the practical complexity.
Therefore, we expect a lower practical complexity than the theoretical one.

\subsection{Results}
We evaluate the scalability of the AFD synthesis algorithm w.r.t to the 3 parameters of our random matrix generators: number of variables, number of configurations and maximum domain size.
In the following, we present the results as figures. When applicable, we include average values using red points, and linear regression trendlines using black lines.
To illustrate how the linear regression fits the data, we also compute the pearson correlation coefficient.

\subsubsection{Or groups}
\label{sec:experiment_or}
As shown in Section~\ref{sec:theoretical}, the computation of or-groups is NP-complete and represents one of the hardest parts of the synthesis algorithm. As such, it may be time consuming and form the main performance bottleneck of the algorithm.
Our first experiment evaluates the scalability of the or-groups computation w.r.t the number of features. We measure the time needed to compute the or-groups from a matrix with 1000 configurations, a maximum domain size of 10 and a number of variables ranging from 5 to 70. To keep a reasonable time for the execution of the experiment, we set a timeout at 30 minutes. Results are presented in Figure~\ref{fig:scalability_or}.

\begin{figure}
\centering
\includegraphics[scale=0.4]{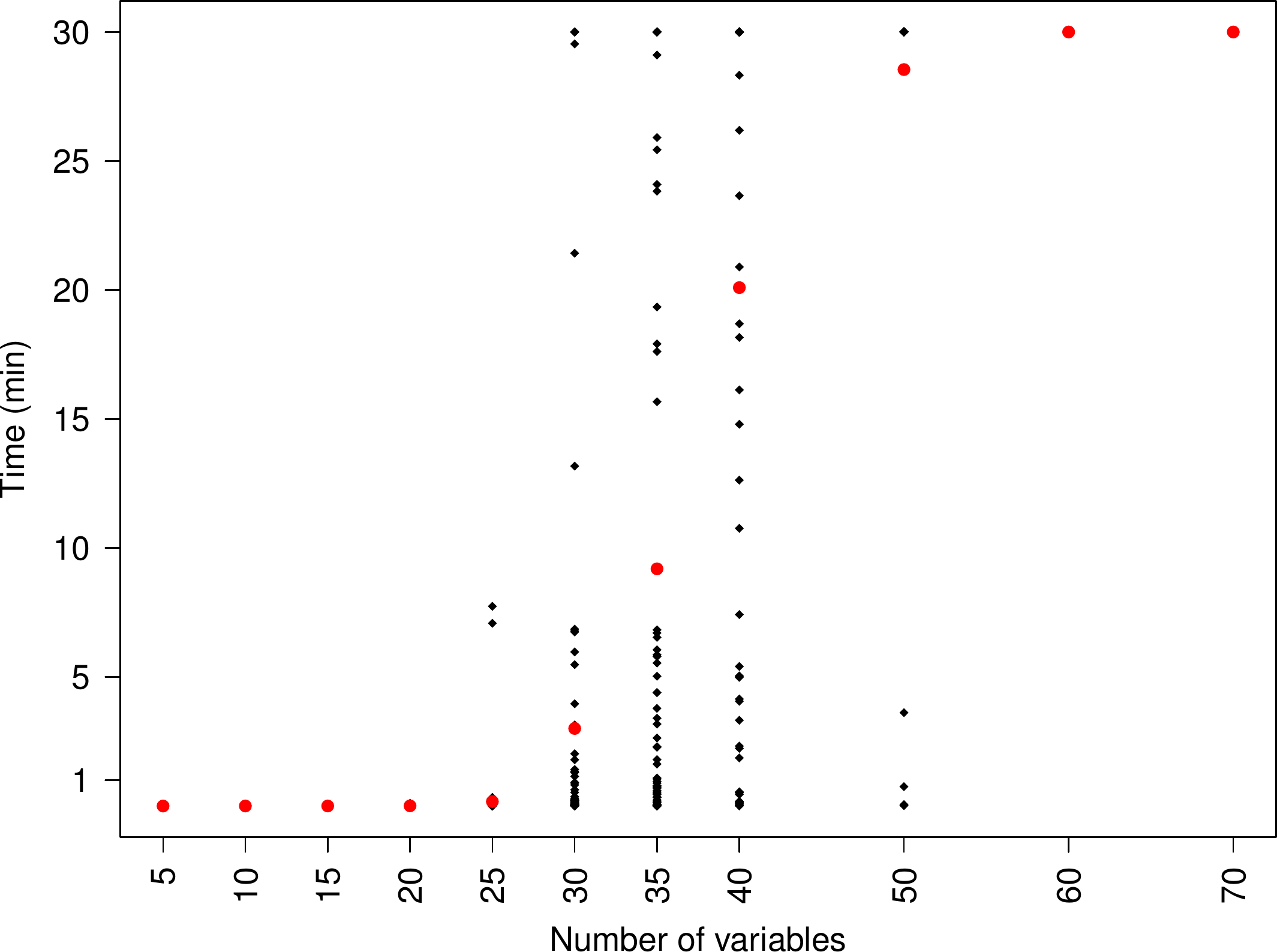}
\caption{\label{fig:scalability_or}Scalability of or-groups computation w.r.t the number of variables}
\end{figure}

The results show that the computation of or-groups quickly becomes time consuming. The 30 minutes timeout is reached with matrices containing only 30 variables. With at least 60 variables, the timeout is systematically reached.
Actually, or-groups represent a small part of the AFD. For example, only one or-group was found in the Linux feature model~\cite{she2011}.
Dedicating so much time for such a small contribution to the feature model is not worth it.
Therefore, we deactivated the computation of or-groups in the following experiments.

\subsubsection{Scalability w.r.t the number of variables}
To evaluate the scalability with respect to the number of variables, we perform the synthesis of random matrices with 1000 configurations, a maximum domain size of 10 and a number of variables ranging from 5 to 2000. 
%In Section~\ref{sec:theoretical}, we showed that if we do not consider or-groups, the time complexity of the synthesis algorithm is related to the square of the number of variables.
In Figure~\ref{fig:scalability_variables}, we present the square root of the time needed for the whole synthesis compared to the number of variables.

\begin{figure}
\centering
\includegraphics[scale=0.4]{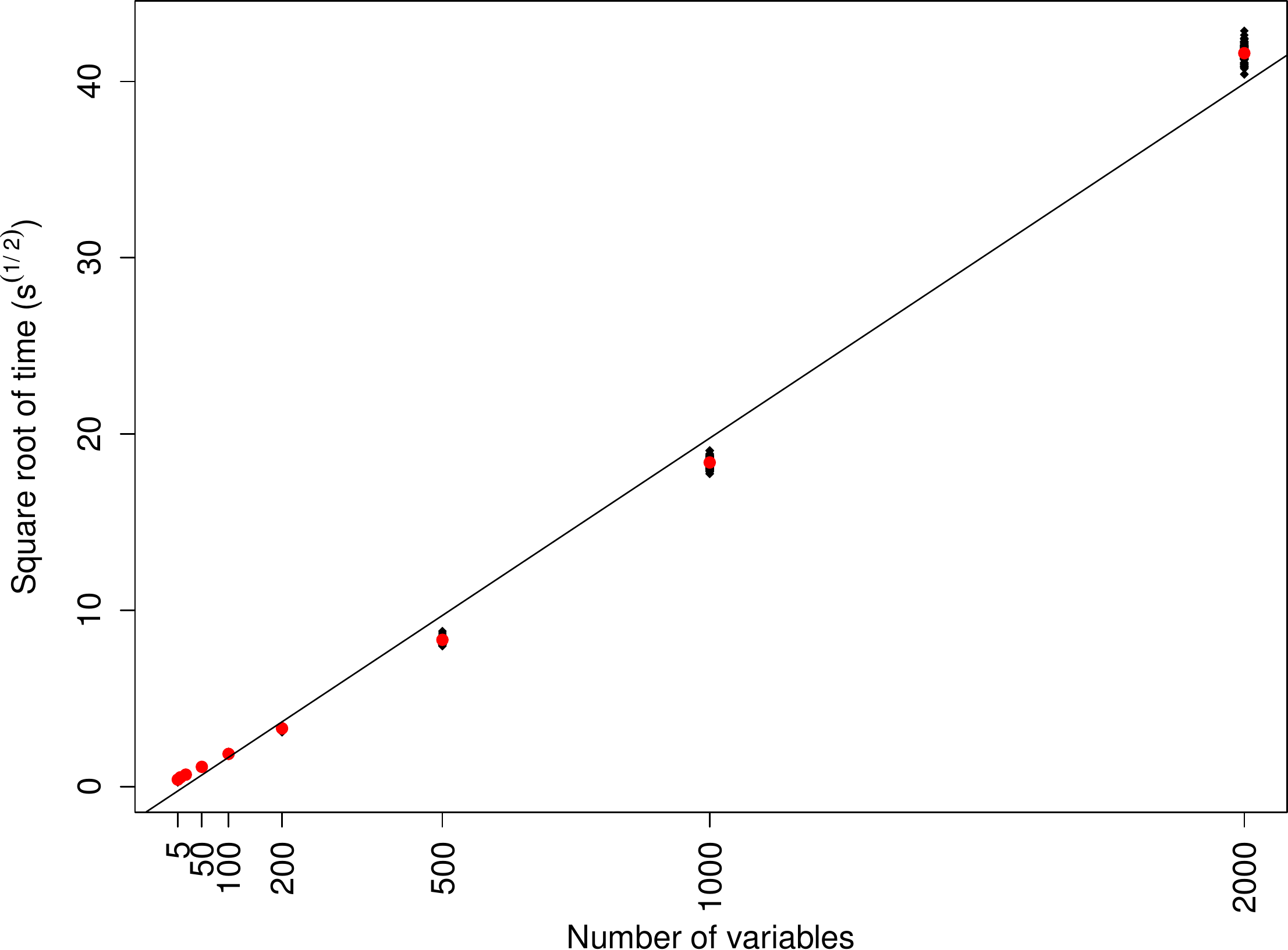}
\caption{\label{fig:scalability_variables}Scalability w.r.t the number of variables}
\end{figure}

The results indicate that the square root of the time grows linearly with the number of variables, with a correlation coefficient of 0.997. 
We observe a difference between pratical and theoretical complexity.
As mentioned in Section~\ref{sec:expected_complexity}, several factors can explain this lower complexity.
Overall, the results of this experiment are consistent with the computed theoretical complexity. 
%However, the theoretical complexity suggests that the time grows proportional to the fourth power of the number of variables. 
%This difference is rooted in the computation of the complex constraints. 
%In our theoretical computation in Section~\ref{sec_CompRC}, we have considered the cost of looking up a binary implication in a map to be $O(v^2)$. 
%This is the worst case time complexity of searching an item in a hash table with $v^2$ elements. 
%However, in practice, we start with empty maps, and gradually add items to them. 
%Therefore, the practical complexity of these look ups throughout the whole process is on average $O(1)$. 
%This implies a quadratic growth of time with the number of variables, which is consistent with the results reported in this section.

\subsubsection{Scalability w.r.t the number of configurations}
To evaluate the scalability with respect to the number of configurations, we perform the synthesis of random matrices with 100 variables, a maximum domain size of 10 and a number of configurations ranging from 5 to 200,000. 
%These values for the number of variables and the maximum domain size ensure that our generator produces a number of disctinct configurations close enough to the parameter of the generator.
With 100 variables, and 10 as the maximum domain size, we can generate $10^{100}$ distinct configurations. This number is big enough to ensure that our generator can randomly generate 5 to 200,000 distinct configurations. 
The time needed for the synthesis is presented in Figure~\ref{fig:scalability_configurations}.

\begin{figure}
\centering
\includegraphics[scale=0.4]{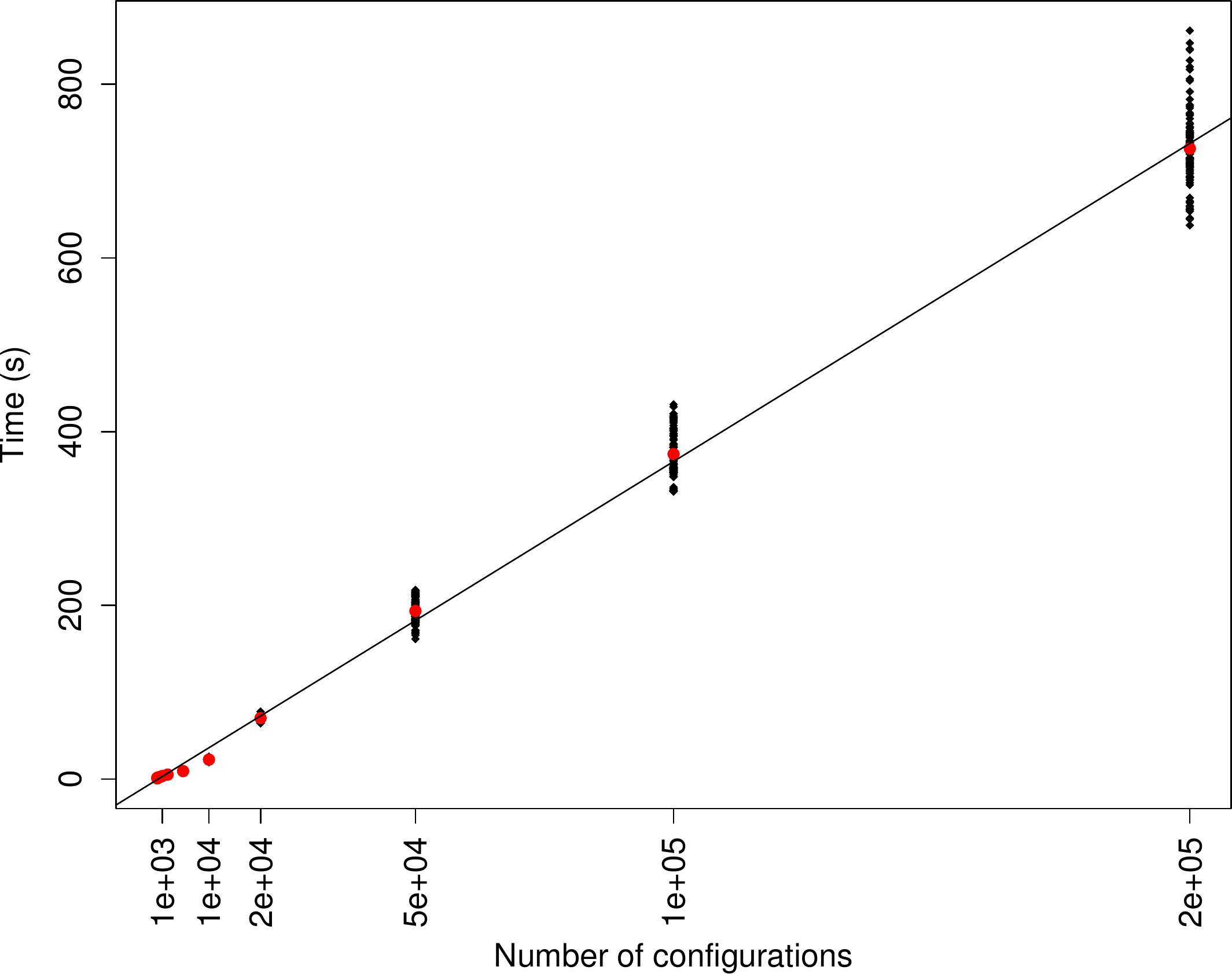}
\caption{\label{fig:scalability_configurations}Scalability w.r.t the number of configurations}
\end{figure}

The results indicate that the time grows linearly with the number of configurations. 
Again, the correlation coefficient of 0.997 confirms that the practical time complexity is consistent with the theoretical complexity.
%\gb{Again, the correlation coefficient of 0.997 confirms that the practical time complexity is consistent with the expected one.}

\subsubsection{Scalability w.r.t the maximum domain size}

\begin{table}
\centering
\small
\begin{tabular}{|c|c|c|c|c|c|c|c|c|c|c|c|}
\hline 
\textbf{Parameter} & 5 & 10 & 20 & 50 & 100 & 200 & 500 & 1,000 & 2,000 & 5,000 & 10,000 \\ 
\hline 
\textbf{Generated} & 5 & 10 & 20 & 50 & 100 & 200 & 500 & 1,000 & 1,995 & 4,385 & 6,416 \\ 
\hline  
\end{tabular} 
\caption{\label{tab:max_domain_size}Difference between the intended maximum domain size and the effective value}
\end{table}

To evaluate the scalability wit respect to the maximum domain size, we perform the synthesis of random matrices with 10 variables, 10,000 configurations and a maximum domain size ranging from 5 to 10,000. The first row of Table~\ref{tab:max_domain_size} lists the intended values for the maximum domain size, used as input to the matrix generator. The second row of the table shows the effective maximum domain size in each case. 
%In Section~\ref{sec:theoretical}, we showed that if we do not consider or-groups, the time complexity of the synthesis algorithm is related to the square of the maximum domain size.
In Figure~\ref{fig:scalability_domain_size}, we present the square root of the time needed for the synthesis.

\begin{figure}
\centering
\includegraphics[scale=0.4]{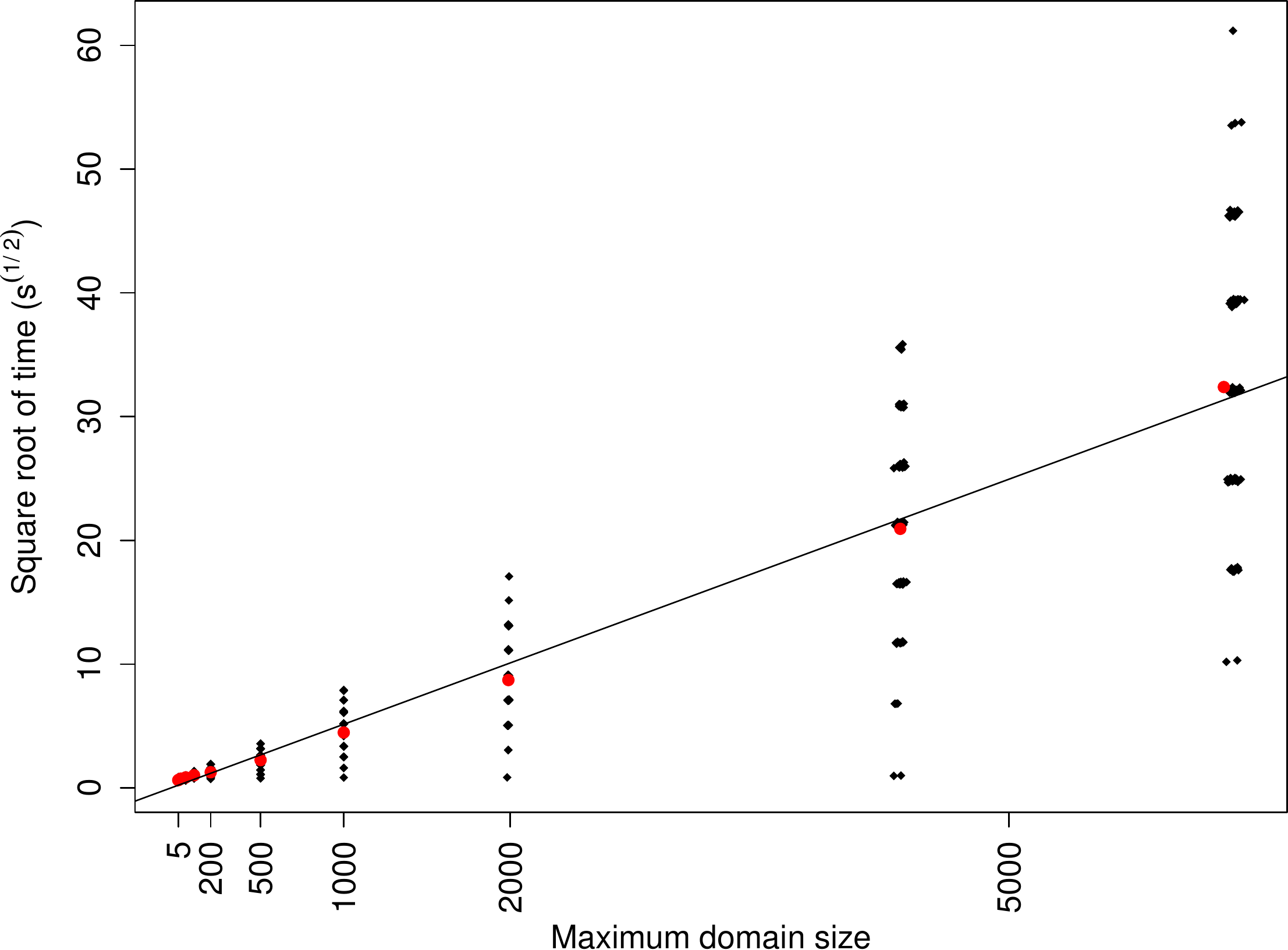}
\caption{\label{fig:scalability_domain_size}Scalability w.r.t the maximum domain size}
\end{figure}

As shown in the second row of Table~\ref{tab:max_domain_size}, and depicted in Figure~\ref{fig:scalability_domain_size}, maximum domain sizes 5000 and 10,000 were not reached during the experiment. 
%As mentioned in Section~\ref{subsec:settings}, the random matrix generator do not ensure that the maximum domain size is reached. The second line of Table~\ref{tab:max_domain_size} shows what are the effective maximum domain sizes of our experiments. 
This does not impact the validity of our results as we show the value of the domain size for the output of the generator and not for its parameters.

We also notice that for each value of the domain size, the points are distributed in small groups. For instance, we can see nine groups of points for a maximum domain size of 2000. 
Each group represents the execution of our algorithm with matrices that have the same number of attributes.
However, we see that the number of attributes does not significantly affect the maximum domain size (the maximum domain size is approximately the same for all groups of results). 
%Therefore, this phenomenon does not impact the validity of our results.

%This phenomenon is also due to our matrix generator. One of its parameters is the number of variables to generate. However, the type of the variable (feature or attribute) is randomly selected. As we are considering the domain value of attributes in this experiment, the performances of our algorithm heavily depends on the number of attributes. Therefore, these groups of points result from the actual difference in the number of attributes contained in the generated matrices. 
%In average, the effective number of attributes is 5. As we are evaluating the scalability w.r.t the maximum domain size, this particular phenomenon is not impacting the validity of our results.

The results indicate that the square root of the time grows linearly with the maximum domain size. 
The correlation coefficient of 0.932 confirms that the practical time complexity is consistent with the theoretically computed time complexity.
%\gb{The correlation coefficient of 0.932 confirms that the practical time complexity is consistent with the expected one.}

\subsubsection{Time Complexity Distribution}

To further understand the practical time complexity of the algorithm, we analyze its distribution over different steps of the algorithm. In Figure~\ref{fig:time_distrib}, we depict the average distribution for all previous experiments that do not contain the computation of or-groups.

The results clearly show that the major part of the algorithm is spent on the computation of binary implications, which has a complexity of $O(v^2.c.d)$ as shown in Section~\ref{sec_CompBI}, and complex constraints for $RC$, which has a complexity of $O(v^4.d + v^2.d^2)$ as shown in Section~\ref{sec_CompRC}. % (respectively steps \ref{ln_bicomp} and \ref{alg:constraints} of Algorithm~\ref{alg_synthesis}). 
The rest of the synthesis represents less than 10\% of the total duration.
Optimizing these two main steps would significantly decrease the time necessary for synthesizing an AFM.

\begin{figure}
\centering
\includegraphics[scale=0.4]{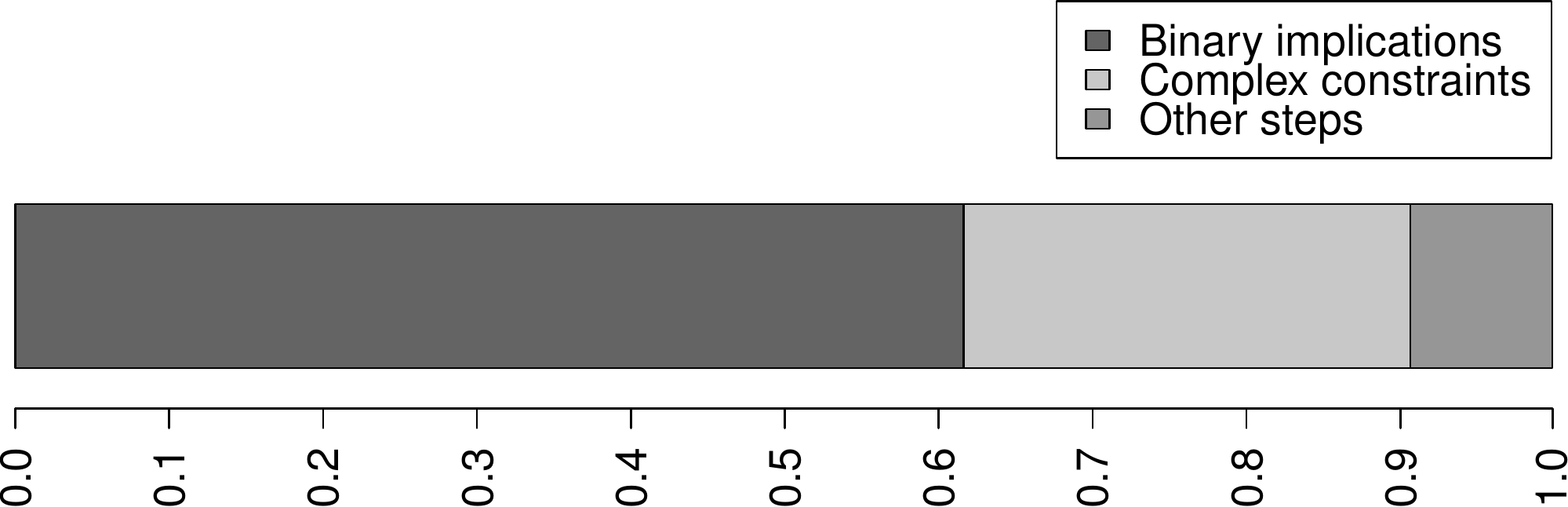}
\caption{\label{fig:time_distrib}Time complexity distribution for all experiments without or-groups computation}
\end{figure}

\section{Threats to Validity} 
\label{sec:threats}
%To perform all these experiments in a reasonable time, we used a cluster of 64 identical nodes.
%All these nodes are composed of 2 Intel Xeon L5420 at 2.5 Ghz with 32GB of RAM.
% paradent cluster : 2 * Intel Xeon L5420 2.5 Ghz / 6MB with 32GB of RAM

% \gb{Can we use our formalism?}
% Yes, it is close to Fama and other academic formalisms
% About complex constraints, we leave it a future work

% \gb{Manual complexity analysis}

An external threat is that our evaluation is based on the generation of random matrices. Using such matrices may not reflect the practical complexity of our algorithm with a realistic data set. We leave this kind of evaluation as future work.
%However, we believe that in some cases, random input can produce worst case scenario.

Evaluating the scalability on a cluster of computers instead of a single one may impact the scalability results and is an internal threat to validity. 
We limited this threat to validity by using a cluster composed of identical nodes. 
Even if the nodes do not represent a standard computer, we only modify the absolute values of the experiments. The practical complexity of the algorithm is not influenced by this gain of processing power. Moreover, all the necessary data for the experiments are present in the local disks of the nodes thus avoiding any network related issue.
Finally, we performed 100 runs for each set of parameters in order to reduce any variation of performance.

Another threat to internal validity is related to our implementation of the algorithm. 
To check the correctness of the implementation, we have manually reviewed some resulting AFMs. We also tested the algorithm against a set of manually designed configuration matrices. Each matrix represents a minimal example of a construct of an AFM (\eg one of the matrices represents an AFM composed of a single xor-group). The test suite covers all the concepts in an AFM. None of these experiments revealed any anomalies in our implementation.

% \gb{Other parameters for the evaluation could lead to other results}
% \gb{Stats}

\section{Conclusion}
\label{sec:conclusion}
We presented the foundations for synthesizing attributed feature models (AFMs) from product descriptions. 
We introduced the formalism of \emph{configuration matrix} for documenting a set of products along different Boolean and numerical values. 
We then sought to understand the relationship between configuration matrices and AFMs. The key contributions of the report can be summarized as follows:
\begin{itemize}
\item We described formal properties of AFMs  and established semantic correspondences with the formalism of configuration matrices. We demonstrated why an attributed feature \textit{diagram} (i.e., the diagrammatic representation of an AFM typically read and maintain by a human) may represent an overapproximation of the configuration matrix; 
\item We designed and implemented a comprehensive, parametrizable synthesis algorithm. We showed that the algorithm is (1) sound and complete w.r.t the configuration semantics of the input configuration matrix (2) synthesizes maximal feature diagrams; %set of synthesized binary constraints is sound and complete by construction and polynomial in time; %... \ma{approximation and equivalence properties: also contribution};
% \item
\item We theoretically and empirically evaluated the scalability of the synthesis algorithm. 
% In particular, we showed that the algorithm has an overall polynomial time complexity, if we do not compute OR-groups.
% \item We empirically investigate the overapproximation effect induced by the formalism of attributed feature model
\end{itemize}

Numerous kinds of artefacts and problems are amenable to configuration matrices, opening perspectives for applying the synthesis of AFMs. 
As future work, we plan to further study some properties of the synthesis -- like scalability (performance), readability and usefulness of computed constraints, and the over-approximation effect. We are currently investigating the use of AFM synthesis in practical settings. 

% We plan to address 
% New research directions emerged in this paper  the possible overwhelming of constraints,  

% \ma{Maybe we have to split the evaluation in two or three sections}
% Since~\cite{czarnecki2007},

\section*{Acknowledgements}
The second author is funded by the Research Council of Norway (the ModelFusion Project - NFR 205606).

Experiments presented in this report were carried out using the Grid'5000 experimental testbed, being developed under the INRIA ALADDIN development action with support from CNRS, RENATER and several Universities as well as other funding bodies (see https://www.grid5000.fr).

\bibliographystyle{IEEEtran}
\bibliography{ICSE15}

\end{document}